\documentclass[12pt]{article}
\usepackage{mathtools,amssymb}
\usepackage{tikz}
\usepackage{amsmath}
\usepackage{amsthm}
\usepackage{geometry}
\usepackage{tgpagella}
\usepackage[title]{appendix}
\usepackage[round]{natbib}
\usepackage{hyperref}
\usepackage{setspace}
\usepackage{float}
\usepackage{times}
\usepackage{subcaption}

\newtheorem{defn}{Definition}
\newtheorem{prop}{Proposition}
\newtheorem{exmp}{Example}
\newtheorem{fact}{Fact}
\newtheorem{thm}{Theorem}
\newtheorem{claim}{Claim}
\newtheorem{lemma}{Lemma}

\newtheorem{coro}{Corollary}
\newtheorem{assu}{Assumption}

\title{Screening Workers with Affirmative Action}
\author{Charles Po-Cheng Huang\footnote{Department of Economics, National Central University, No. 300, Zhongda Rd., Zhongli District, Taoyuan City, Taiwan 320317. Email: cpchuang@ncu.edu.tw. I am especially grateful to Bumin Yenmez, Rahul Deb, and Utku Ünver for their advising and their continuous support. I thank Min-Hung Tsay, Alice Peng-Ju Su, CJ Sun, İsa Hafalır, Kentaro Tomoeda, Gunhaeng Lee, and Johan Lagerlöf for the suggestion and comments. I also thank the participants in RCHSS at Academia Sinica, Taiwan Economic Research (2025), 20th EGSC at WashU, Boston College Dissertation Workshop, SSCW 2026 for their feedback.}}
\date{July, 2026}

\begin{document}    
	\maketitle
	
	%%%%%%%%%%%%%%%%%%%%%%%%%%%%%%%%%%%%%%%%%%%%%%%%%%%%%%%%%%%%%%%%%%%%%%%%%%%%%%%%%%%%%%%
	%%%%%%% Abstract %%%%%%%%%%%%%%%%%%%%%%%%%%%%%%%%%%%%%%%%%%%%%%%%%%%%%%%%%%%%%%%%%%%%%%
	%%%%%%%%%%%%%%%%%%%%%%%%%%%%%%%%%%%%%%%%%%%%%%%%%%%%%%%%%%%%%%%%%%%%%%%%%%%%%%%%%%%%%%%
	\begin{abstract}
		This paper examines the optimal contracts in a two-dimensional screening model where one dimension (group identity) is verifiable by agents but not falsifiable. A principal offers contracts to agents who differ in cost types and group membership. Motivated by the United States Federal policy, Work Opportunity Tax Credit (WOTC), the principal receives tax benefits for hiring agents from protected groups. Under the assumption that the protected agents tend to have higher cost types, the optimal contract induces full separation across both dimensions: agents reveal the cost type and the group identity through contract choice. Furthermore, the principal is willing to hire the trait agents up to a higher cost threshold than the non-trait agents, and this threshold increases with the tax credit. Conversely, when the protected agents tend to have lower cost types, the optimal design without tax credits pools groups while separating by cost type. These results demonstrate that both affirmative action and non-discrimination can be optimal depending on the cost distribution ordering across groups.

		\vspace{1cm}
	\end{abstract}
	
	%%%%%%%%%%%%%%%%%%%%%%%%%%%%%%%%%%%%%%%%%%%%%%%%%%%%%%%%%%%%%%%%%%%%%%%%%%%%%%%%%%%%%%%
	%%%%%%% Introduction %%%%%%%%%%%%%%%%%%%%%%%%%%%%%%%%%%%%%%%%%%%%%%%%%%%%%%%%%%%%%%%%%%%%%%%%%%%%
	%%%%%%%%%%%%%%%%%%%%%%%%%%%%%%%%%%%%%%%%%%%%%%%%%%%%%%%%%%%%%%%%%%%%%%%%%%%%%%%%%%%%%%%
	\newpage
	\section{Introduction}
	
	In labor markets where worker diversity is considered, firms face screening challenges when designing employment contracts. While traditional contract theory has extensively examined ability-based screening, the interaction between ability and diversity backgrounds remains less explored. This paper examines how an employer optimally designs contracts when workers differ along two dimensions: their production abilities and their backgrounds, which may place them in a protected group that qualifies the employer for tax incentives.

	Affirmative action has been used to protect different minority groups in the United States for many decades. It initially aimed to level the playing field for African Americans, who have experienced a long history of oppression; afterward, it expanded to include protections for individuals from low socioeconomic backgrounds. The goal of affirmative action policies is to create equal competition regardless of one's background. Evidence shows that the United States (US) Federal government also promotes equality and diversity in the workplace. Consider this quote from the website of the US Office of Federal Contract Compliance Programs:

	\begin{quotation}
		The purpose of affirmative action is to ensure equal employment opportunities for applicants and employees. It is based on the premise that, absent discrimination, over time a contractor's workforce generally will reflect the demographics of the qualified available workforce in the relevant job market. Affirmative action requirements are intended to ensure that applicants and employees of federal contractors have equal opportunity for recruitment, selection, advancement, and every other term and privilege associated with employment, without regard to their race, color, religion, sex, sexual orientation, gender identity, national origin, disability, or status as a protected veteran.
	\end{quotation}

	A US Federal government policy called the ``Work Opportunity Tax Credit (WOTC)'', administered by the Internal Revenue Service and the Department of Labor, serves a similar purpose, helping targeted groups to gain access to employment. This tax credit is available for employers who hire employees from targeted groups such as veterans, IV-A recipients,\footnote{A Qualified IV-A recipient is an individual in a family receiving Temporary Assistance for Needy Families (TANF) for at least nine months within the 18-month period ending on their hiring date.} ex-felons, and others. The WOTC is about 40\% of wages (up to \$6,000) paid to a newly hired individual who is certified as a member of a targeted group and works full-time.\footnote{See \cite{collins2018work} for detailed descriptions and the summary statistics of this federal policy. Furthermore, the policy itself can be found on the Internal Revenue Service website and the performance of the WOTC can be found on the Department of Labor website. The cap of the tax credit varies with the certified group a worker belongs to.}

	This paper considers a mechanism design framework involving one principal (she) and two groups of agents. We pose the following questions: In a hiring scenario where agents differ in abilities and backgrounds, how would the principal design contracts to effectively screen the agent? If the contracts can screen the agent, is he screened by ability, background, or both? Moreover, what insights do the optimal contracts provide? The principal's task is to design a menu of contracts, from which the agents will choose their most preferred contract. We model an agent's type as two-dimensional. The first dimension is the innate ability, captured by a privately known cost type, and the second dimension is the group he belongs to. The group determines the distribution from which the agent draws this cost type. Furthermore, we consider binary groups: a trait group that includes all protected identities and a non-trait group. Therefore, being in the trait group means the agent has the identity that allows the principal to claim the tax credit.

	This model incorporates a key asymmetry in information revelation: an agent in the trait group has verifiable evidence of his membership and can decide whether to present it, while a non-trait agent cannot pretend to be in the trait group. Therefore, the non-trait agent can only choose the contracts designed for the non-trait group, but the trait agent can choose any contract offered by the principal. This affirmative action structure creates asymmetric incentives between the groups. The principal knows this asymmetric structure and the proportion of each group, but she knows neither the exact cost type nor the group the agent belongs to.

	Our first main result (part 1 of Theorem~\ref{thm:main}) characterizes the optimal contract under the empirically plausible assumption that the trait group tends to have higher cost types, which is captured formally by the Monotone Likelihood Ratio Property (MLRP). This assumption reflects that the trait agents often face labor market discrimination or other disadvantages that affirmative action policies aim to address. Under MLRP, we show that the optimal contract design induces full separation: agents reveal both their cost type and group identity through their contract choices. Even though the trait agents could choose contracts designed for the non-trait agents, they prefer to disclose their trait status. Furthermore, the principal is willing to hire the trait agents up to a higher cost threshold than the non-trait agents, and this threshold increases with the tax credit.

	Our second main result (part 2 of Theorem~\ref{thm:main}) establishes the optimal contract when the reverse ordering holds: the trait agents tend to have lower cost types (reverse MLRP). Consider first the case without the tax credit. Under reverse MLRP, the optimal contract design pools the groups and separates the agents by cost type: the principal offers identical contracts based solely on cost type, making group identity irrelevant. This group-blind design reflects a non-discrimination policy, which bans the firm from treating workers with the same ability differently. When the tax credit appears, we show that the optimal contracts induce full separation for the low-cost segment and pool the groups for the high-cost segment. In other words, the principal finds it worthwhile to screen for the tax credit only among the lower-cost trait agents.
	
	We develop the policy content of this dichotomy in Section~\ref{sec:policy}. The same tax credit endogenously produces either explicit group-targeting, which is affirmative action, or de facto non-discrimination. The ordering of the two cost distributions, rather than the policy itself, determines which regime obtains. This observation bears directly on the empirical literature on the WOTC, which consistently finds that the subsidy produces little visible change in employers' hiring standards and contracts. Our model supplies a frictionless benchmark for interpreting these null results: under reverse MLRP, the optimal contracts leave no visible trace of the tax credit over the pooling region, so a null finding is consistent with optimal screening and does not by itself indicate frictions. We further argue that both orderings are empirically plausible: MLRP reflects the barriers to employment that define the target groups, while the applicant pools facing the firms that actually claim the credit are heavily pre-screened and can satisfy the reverse ordering. Moreover, two seemingly convenient modeling choices, group-specific menus and one-way disclosure, turn out to match the legal and certification design of the WOTC. Finally, the disclosure option has a substantive incidence implication: because a trait agent always retains the option to take the contract designed for the non-trait agents, the tax credit can never make him worse off through the contracting channel, in contrast to models in the tradition of \cite{coate1993will} where affirmative action can harm its intended beneficiaries.

	Our analysis also makes a methodological contribution to mechanism design. The information structure we study, a privately known cost type coupled with a second dimension that agents can verify but not falsify, has been examined in the literature on evidence, but existing characterizations rely on the principal's objective being linear in the allocation (\cite{vaidya2023regulating}). In our setting, production is continuous and the cost function is strictly convex, so those arguments are unavailable. We instead use the calculus of variations, which yields a deterministic optimal contract in which the binding disclosure constraint is resolved by pooling the groups. This characterization is what delivers the dichotomy above, making transparent when affirmative action and when non-discrimination emerge as optimal from the principal's objective rather than being imposed by regulation.
	
	\subsection*{Related Literature}
	Our paper builds on the monopolistic screening literature. The seminal work of \cite{mussa1978monopoly} analyzed how a monopolist should design quality-differentiated products when consumers have private information on valuations for quality. \cite{maskin1984monopoly} extended this work by allowing for a general preference structure and developed a systematic technique to solve the screening problems. While these papers provided the methodology used elsewhere in this literature, our work investigates the possibility of agents having group identity and cost at the same time, which brings us to a multidimensional problem. A related strand of the screening literature studies contracts when the agent's participation constraint is type-dependent, so that the outside option varies with the agent's private information. \cite{jullien2000participation} provides the canonical treatment of this problem and shows that the optimal contract can feature nonmonotone information rents and binding participation constraints at interior types. In our model, the trait agent's outside option is genuinely type-dependent: it varies with the non-trait agent's contract. Our analysis differs in that this type-dependent outside option arises endogenously from the verifiability structure rather than from an exogenously specified reservation utility function.

    \cite{mcafee1988multidimensional} gave conditions under which multidimensional screening reduces to a one-dimensional problem, while \cite{armstrong1999multi} and \cite{rochet2003economics} illustrated the difficulty of adding even a single dimension. Our second dimension has a particular structure that keeps the problem tractable: group identity is binary, does not enter the utility function directly but governs only the distribution from which the cost type is drawn, and, crucially, can be misrepresented in only one direction. Our pooling result also connects to a recent literature identifying distributional conditions under which optimal menus collapse in multidimensional environments. \cite{haghpanah2021pure} showed that a multiproduct monopolist optimally sells only the grand bundle precisely when types with higher bundle values are stochastically more favorable to the seller, a condition with the same likelihood-ratio flavor as the reverse ordering that generates pooling in our model.
    
	The verifiability of group identity connects to the literature on verifiable evidence in mechanism design. Starting from \cite{green1986partially}, a series of papers studied how verifiable information affects the mechanism that a social planner would design and the implementability of mechanisms (e.g., \cite{bull2004evidence}, \cite{bull2007hard}, \cite{deneckere2008mechanism}, \cite{ben2019mechanisms}). Our restriction to direct mechanisms in which agents disclose all available evidence is justified by the normality condition of \cite{bull2007hard}, which our binary evidence structure satisfies: an agent who can prove membership in the protected group can equally withhold that proof. We bring the verifiable evidence feature to the screening problem and analyze the optimal quantity that the principal should ask for. 
    
    Within this literature, the following papers are close to ours. \cite{sher2015price} study revenue-maximizing sale of a single good to a buyer with finitely many types holding overlapping sets of evidence, and characterize the optimal mechanism through linear-programming duality on an incentive graph. Our two-group structure is the simplest nontrivial case of their graph, but our allocation problem, which features continuous production, a strictly concave objective, and a continuous subsidy, lies outside their finite linear-programming framework, whose endogenous duality characterization is, as they note, not suited to comparative statics. \cite{vaidya2023regulating} is closest in how evidence operates: his buyer has a continuous payoff type and finitely many evidence realizations, and his incentive constraints decompose exactly as ours do, into standard constraints within each evidence class and a unidirectional disclosure constraint across classes. His characterization, however, exploits the linearity of the seller's objective in the allocation, identifying the optimum as an extreme point of the feasible set that randomizes over posted prices when the disclosure constraint binds. With a strictly convex production cost, this approach is unavailable. His own convex-cost extension retreats to a Dye evidence structure in which each payoff type holds uniquely identifying evidence, eliminating precisely the two overlapping populations we study. We solve this configuration under a general convex cost using the calculus of variations, and the binding disclosure constraint is resolved by deterministic pooling rather than randomization. \cite{krahmer2025unidirectional} study unidirectional incentive compatibility in the canonical one-dimensional screening model, showing it removes the monotonicity restriction on allocations and breaks revenue equivalence; in our model, unidirectionality instead operates along the binary verifiable dimension, and its interaction with the standard bidirectional cost dimension is what generates the dichotomy between separation and pooling. Finally, \cite{pram2023learning} also studies screening with evidence, in an insurance market where the evidence concerns the agent's own risk type rather than group membership.
    
	As we discuss the Work Opportunity Tax Credit (WOTC) that allows the principal to receive tax benefits by hiring an agent from the protected group, our work also relates to the literature on affirmative action. \cite{ajilore2012did}, \cite{hamersma2003work}, \cite{gunderson2007job}, and \cite{jain2026limits} empirically studied the effect of the WOTC alongside similar programs. \cite{coate1993will} analyzed how affirmative action policies affect the equilibrium outcome in the labor market, and \cite{fryer2005affirmative} studied different approaches to affirmative action and compared the consequences. \cite{chan2003does} adopted the mechanism design approach to study the effect of affirmative action in college admissions. In a similar vein, \cite{passaro2023equal} examined equal pay for similar work policies, finding they can induce workforce segregation and unexpectedly widen the wage gap. \cite{dianat2022statistical} studied affirmative action and discrimination in an experimental setting and found that affirmative action that comes as a subsidy to the firm tends to have a short-term effect on reducing statistical discrimination. \cite{ho2020efficient} designed efficient child-care subsidies as a screening problem, an exercise close to ours in spirit in that a policy instrument is embedded in an incentive problem with asymmetric constraints. In the market design literature, \cite{abdulkadirouglu2003school}, \cite{hafalir2013effective}, and \cite{sonmez2022affirmative} studied the allocation problem under diversity constraints. While these papers focus on developing allocation mechanisms or equilibrium, our main goal is to study how a principal designs contracts for each type when diversity becomes another consideration and agents can choose not to present the evidence.  
	% Cite Joesph's working paper as another empricial finding for AA.

    To summarize, our work extends the screening literature by adding a dimension whose types are partially verifiable and by incorporating diversity incentives into the principal's objective. The framework allows us to characterize optimal contracts under a novel verifiability constraint and to derive comparative statics describing how the protected group is treated as the tax benefit varies.
	
	The remainder of this paper is organized as follows. Section \ref{sec:model} presents our two-dimensional screening model with verifiable group evidence. Section \ref{sec:optimal} characterizes the optimal contracts under both MLRP and reverse MLRP settings. Section \ref{sec:quadratic} explicitly shows the optimal quantities and wages under a quadratic cost function. Section \ref{sec:policy} draws out the policy implications of our results, relating the dichotomy between separation and pooling to the empirical evidence on the WOTC and to the program's institutional design. Section \ref{sec:conclusion} concludes with directions for future research. All the proofs are left to appendix \ref{app:pfs}.
	
	\section{Two-Dimensional Screening Model}\label{sec:model}
	%%%%%%%%%%%%%%%%%%%%%%%%%%%%%%%%%%%%%%%%%%%%%%%%%%%%%%%%%%%%%%%%%%%%%%%%%%%%%%%%%%%%%%%
	%%%%%%% Model  %%%%%%%%%%%%%%%%%%%%%%%%%%%%%%%%%%%%%%%%%%%%%%%%%%%%%%%%%%%%%%%%%%%%%%%%%
	%%%%%%%%%%%%%%%%%%%%%%%%%%%%%%%%%%%%%%%%%%%%%%%%%%%%%%%%%%%%%%%%%%%%%%%%%%%%%%%%%%%%%%%
	We consider a screening problem with one principal (she) and a continuum of agents (he). In this context, the principal is the company trying to design contracts for the workers (agents). An agent has a type $(\theta, j)$, where $\theta \in [0,1]$ represents the cost type of this agent, and $j \in \{N,T\}$ captures whether the agent belongs to the trait group ($j = T$) or not ($j = N$). The proportion of the trait group is $\pi_T$, and the proportion of the non-trait group is $1-\pi_T$.  For a trait agent, whose type is $(\theta, T)$, the cost type $\theta$ is drawn from a CDF $F^T$, where we assume the PDF exists, and we denote it as $f^T$. Similarly, for a non-trait agent $(\theta, N)$, his cost type $\theta$ is drawn from another CDF $F^N$, whose PDF exists and is denoted by $f^N$. We let $F^T$ and $F^N$ have the same support: $[0,1]$, which is a normalization. 
	
	A contract is a pair $(w,q)$, a wage term and a production term. The production is captured by the quantity asked to produce, and we require both the wage term and the quantity to be nonnegative. Hence, the payoff of an agent of type $\theta$ from contract $(w,q)$ is 
	$$U\big(w, q; \theta \big) = w - C(q,\theta),$$
	where $C(\cdot)$ is a cost function mapping from $\mathbb{R} \times [0,1]$ to $\mathbb{R}$.	We assume $C$ is twice continuously differentiable, with $\frac{\partial C(\cdot)}{\partial q} > 0$, $\frac{\partial C(\cdot)}{\partial \theta} > 0$, $\frac{\partial^2 C(\cdot)}{\partial q^2} > 0$, and $\frac{\partial^2 C(\cdot)}{\partial q \partial \theta} > 0$.\footnote{This is the Spence-Mirrlees single-crossing condition.} Notice that the group identities $N$ and $T$ do not enter the utility function explicitly; they only affect the draw of the cost type $\theta$. 

	For the principal, the proportion of each group and the distributions $F^T$ and $F^N$ are common knowledge, but she does not know the exact cost type $\theta$ or to which distribution an agent belongs. Furthermore, she knows that the trait agent can choose either to present the evidence or hide the evidence when accepting a contract; however, the non-trait agent cannot pretend to be in the trait group, as he does not have the evidence to present when accepting the contract. 

	By the revelation principle of \cite{bull2007hard}, it is without loss of generality to restrict attention to direct mechanisms that request maximal evidence; thus, we focus on the direct mechanisms that (1.) ask the agents to report the cost type, and (2.) request the trait agents to deliver proof of their status. Whether compliance is incentive-compatible is governed by the within-group incentive compatibility constraint and the across-group incentive compatibility constraint, which we will introduce later in Section~\ref{subsec:Agent}. A menu of contracts is thus a pair of functions $\big(w^j(\theta), q^j(\theta)\big)$, where $j \in \{N,T\}$ indexes the group.
	
	When facing an agent with type $(\theta, j)$, the principal's utility is 
	$$V\big(w^j(\theta), q^j(\theta) \big) = P q^j(\theta) - (1-\tau(j) )w^j(\theta),$$
	where $P > 0$ is a given price of selling the product and we assume $P > \frac{\partial C(0, 0)}{\partial q}$ to ensure there is a gain from accepting the contract,  $j \in \{N,T\}$ is the group index, and $\tau$ captures the WOTC. $\tau(T) = \tau \in [0,1)$ is the subsidy or tax credit of hiring a trait agent, while $\tau(N) = 0$ for hiring a non-trait agent. The tax credit is modeled as depending on the agent's true trait status. 
	
	\subsection{Timeline}
	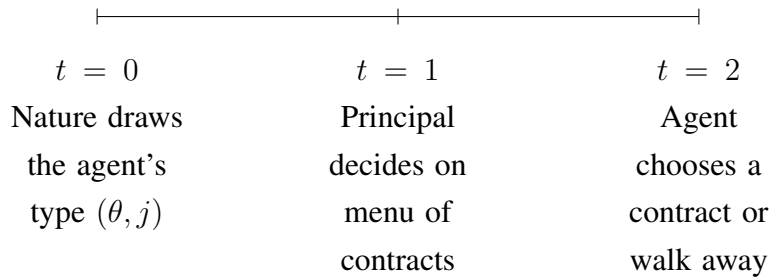
\begin{figure}[H]
		\centering
		\begin{tikzpicture}
			% Draw horizontal line
			\draw (-4,0) -- (4,0);
			
			% Draw vertical lines
			\foreach \x in {-4, 0 ,4}
			\draw (\x cm,3pt) -- (\x cm,-3pt);
			
			% Draw nodes
			\draw (-4,0) node[below=12pt, text width=2.5cm, align=center] 
			{$t=0$ \\ Nature draws the agent's type $(\theta, j)$}; % for agent
			\draw (0,0) node[below=12pt, text width=2.5cm, align=center] 
			{$t=1$ \\ Principal decides on menu of contracts};
			\draw (4,0) node[below=12pt, text width=2.5cm, align=center] 
			{$t=2$ \\ Agent chooses a contract or walks away};
		\end{tikzpicture}
		\caption{Timeline of the Model}
		\label{fig:timeline}
	\end{figure}
	
	At $t = 0$, Nature draws the group status ($N$ or $T$) and the cost $\theta$ for each agent. At $t = 1$, the principal designs the menu(s) of contracts to offer to the agents. At $t = 2$, the agent chooses one contract from the menu(s) or walks away without taking any contract. We normalize the utility of walking away to be zero. If the agent chooses a contract designed for the trait group, he is required to show evidence of being a trait agent. Only the trait agent can decide whether to present his trait status to the principal to claim the contract designed for the trait group. Upon seeing the trait status, the principal knows this agent must be a trait agent.
	
	\subsection{Agent's Problem}\label{subsec:Agent}
	We now consider the incentive compatibility and individual rationality for the agent. For the incentive compatibility, we define the following two concepts:
	
	\begin{defn}[Within-Group Incentive Compatibility]\label{def:WGIC}
		For an agent with cost $\theta \in [0,1]$, $$U(w^j(\theta), q^j(\theta); \theta) \geq U(w^j(\theta'), q^j(\theta'); \theta),$$
		for any $\theta' \in [0,1]$ and $j \in \{N,T\}$.
	\end{defn}
	
	This is the standard incentive compatibility condition, renamed within-group incentive compatibility to distinguish it from the across-group constraint below:
	
	\begin{defn}[Across-Group Incentive Compatibility]\label{def:AGIC}
		For a \textbf{trait} agent with cost $\theta \in [0,1]$, $$U(w^T(\theta), q^T(\theta); \theta) \geq U(w^N(\theta'), q^N(\theta'); \theta),$$
		for any $\theta' \in [0,1]$.
	\end{defn}
	
	This definition captures the across-group incentive. Since only the trait agent can deviate across groups (not vice versa), this constraint is unidirectional. In other words, a trait agent with cost $\theta$ should weakly prefer the contract designed for the trait agent with cost $\theta$ to that for the non-trait agent with any cost $\theta'$.
	
	Definitions \ref{def:WGIC} and \ref{def:AGIC} together give us the full incentive compatibility concept in this model. An agent will truthfully choose the contract designed for his cost $\theta$ and, on top of this, a trait agent will also truthfully choose the contract designed for his group. Therefore, we have incentive compatibility for both dimensions $\theta \in [0,1]$ and $j \in \{N,T\}$.
	
	We define individual rationality in the usual sense:
	\begin{defn}[Individual Rationality]\label{def:IR}
		For an agent with cost $\theta \in [0,1]$, $$U(w^j(\theta), q^j(\theta); \theta) \geq 0,$$
		for $j \in \{N,T\}$.
	\end{defn}
	
	For an agent with cost $\theta$, choosing the contract designed for him must yield at least zero payoff. Otherwise, the agent can always walk away without accepting any contract, and we normalize the outside option to be $0$.

	\subsection{Principal's Objective}
	The principal designs menus of contracts $\{(w^j,q^j)\}_{j \in \{N,T\}}$ to solve the following constrained optimization problem:
	\begin{align*}
		& \pi_T \int^1_0 \bigg[ Pq^T(\theta)- (1-\tau)w^T(\theta) \bigg] dF^T(\theta) + (1-\pi_T) \int^1_0 \bigg[ Pq^N(\theta) - w^N(\theta) \bigg] dF^N(\theta)
	\end{align*}
	\begin{align*}
		\text{s.t.} \quad 
		& \text{ within-group incentive compatibility constraints}, \\
		& \text{ across-group incentive compatibility constraints}, \\
		& \text{ individual rationality constraints}. \\
	\end{align*} 
	
	The first integral is taken over $F^T$, which is the CDF of the trait group, while the second integral is taken over the non-trait group. The integrals are weighted by the shares $\pi_T$ and $1- \pi_T$, respectively.

	\section{Optimal Contracts}\label{sec:optimal}
	%%%%%%%%%%%%%%%%%%%%%%%%%%%%%%%%%%%%%%%%%%%%%%%%%%%%%%%%%%%%%%%%%%%%%%%%%%%%%%%%%%%%%%%
	%%%%%%%%%%%%%%%%%%%% Result %%%%%%%%%%%%%%%%%%%%%%%%%%%%%%%%%%%%%%%%%%%%%%%%%%%%%%%%
	%%%%%%%%%%%%%%%%%%%%%%%%%%%%%%%%%%%%%%%%%%%%%%%%%%%%%%%%%%%%%%%%%%%%%%%%%%%%%%%%%%%%%%%
	We start by analyzing the agent's problem, where we will find necessary and sufficient conditions for the within-group incentive compatibility (IC) constraints, the across-group IC constraints, and the individual rationality (IR) constraints to be satisfied. We then replace the constraints in the principal's problem and solve for the optimal contract design.
	
	\subsection{The Agent's Problem}
	Consider the within-group IC constraints; these are the standard IC constraints that prevent an agent with cost $\theta$ from choosing a contract for another cost type. We can then follow the standard argument from \cite{borgers2015introduction} and give the following result:
	
	\begin{lemma}[Within-Group IC]\label{lemma:withIC}
		A set of contracts $(w^j, q^j)$ is \textbf{within-group incentive compatible} if and only if
		\begin{align*}
			& (1.) \quad q^j(\theta) \text{ is non-increasing in } \theta\\ %monotonicity
			& (2.) \quad w^j(\theta) = w^j(\theta)|_{\theta = 1} + C(q^j(\theta), \theta) - C(q^j(\theta), \theta)|_{\theta = 1} + \int^{1}_\theta \frac{\partial C(q^j(x), x)}{\partial x}dx, % revenue equivalence
		\end{align*} 
		where $j \in \{N,T\}.$
	\end{lemma}

	In words, Lemma \ref{lemma:withIC} gives us a necessary and sufficient condition for a set of contracts $(w^j,q^j)$ to be within-group IC: a monotonicity condition and a revenue equivalence condition. The monotonicity condition says the quantity cannot increase when the agent becomes more costly. The revenue equivalence condition says the wage for any cost type $\theta$ is determined by the wage of the highest cost agent ($\theta = 1$), the difference in direct production cost, and the cumulative partial effect of type on production cost, $\int^1_\theta \frac{\partial C(q^j(x),x)}{\partial x} dx$, which captures the informational rent.
	
	Next, we consider the across-group IC constraints. As the across-group IC constraints allow the trait agent to choose any contract from the non-trait menu, we first analyze which is the best non-trait contract to choose. We have the following lemma:
	
	\begin{lemma}\label{lemma:acrossIC}
		A trait agent with cost $\theta$ pretending to be a non-trait agent with any cost type maximizes his utility by choosing the contract $\big(w^N(\theta), q^N(\theta) \big)$.
	\end{lemma}

	Therefore, together with the within-group IC, the across-group IC reduces to comparing the utility received from the contract for the trait agent to the utility for the non-trait agent under the same cost type $\theta$, that is, $U(w^T(\theta), q^T(\theta); \theta) \geq U(w^N(\theta), q^N(\theta); \theta)$. This result comes from the fact that the group identity does not enter the utility function. We can state the following lemma for the agent's incentive problem:
	
	\begin{lemma}\label{lemma:ICs}
		A set of contracts $(w^j, q^j)$ is \textbf{within-group incentive compatible} and \textbf{across-group incentive compatible} if and only if
		\begin{align*}
			& (1.) \quad \text{for any } \theta, q^j(\theta) \text{ is non-increasing in } \theta, \\ % monotonicity
			& (2.) \quad \text{for any } \theta, w^j(\theta) = w^j(\theta)|_{\theta = 1} + C(q^j(\theta),\theta) - C(q^j(\theta), \theta)|_{\theta = 1} + \int^{1}_\theta \frac{\partial C(q^j(x), x)}{\partial x} dx, \\  % revenue equivalence
			& (3.) \quad \text{for any } \theta, \int^{1}_{\theta} \bigg( \frac{\partial C(q^N(x), x)}{\partial x} - \frac{\partial C(q^T(x), x)}{\partial x} \bigg) dx \leq U\big( w^T(\theta), q^T(\theta); \theta \big)|_{\theta = 1} - U\big( w^N(\theta), q^N(\theta); \theta \big)|_{\theta = 1},
		\end{align*} 
		where $j \in \{N,T\}.$
	\end{lemma}
	
	Thus far, we have the agent's full incentive problem; the rest is to consider the individual rationality constraint. To simplify the notation, we let $U^j(\theta) =  U\big( w^j(\theta), q^j(\theta); \theta \big)$, and we can do so because the within-group incentive constraints are guaranteed by monotonicity and revenue equivalence. Also, let $U^j(1) =  U\big( w^j(\theta), q^j(\theta); \theta \big)|_{\theta = 1}$, for $j \in \{N,T\}$. Lemma \ref{lemma:IR} characterizes the binding individual rationality constraints by establishing the utility of the highest cost agents in both groups.
	\begin{lemma}\label{lemma:IR}
		If a within-group incentive compatible, across-group incentive compatible, and individual rational, set of contracts maximizes the principal's objective, then $$U^N(1) = 0 \quad
		\text{ and } \quad 
		U^T(1) = \max_{\theta \in [0,1]} \int_{\theta}^{1} \bigg( \frac{\partial C(q^N(x), x)}{\partial x} - \frac{\partial C(q^T(x), x)}{\partial x} \bigg) dx. $$
	\end{lemma}

	Lemma \ref{lemma:IR} shows that the most costly non-trait agent must receive zero utility, which is a standard result in the literature; however, the most costly trait agent need not receive zero utility; his utility is determined endogenously and depends on the integral of the difference in the partial effect of type on cost across the two groups. In other words, it depends on the information rent given to the two groups. 
	
	Next, we turn to the principal's problem.
	
	\subsection{The Principal's Problem}
	We can rewrite the principal's problem as follows:
	\begin{align*}
		\max_{q^T, q^N, U^T(1)} & \pi_T \int^1_0 \biggl\{ Pq^T(\theta)- (1-\tau) \bigg[ U^T(1) + \frac{\partial C(q^T(\theta), \theta)}{\partial \theta} \frac{F^T(\theta)}{f^T(\theta)} + C(q^T(\theta), \theta)\bigg] \biggr\} dF^T(\theta) + \\
		&(1-\pi_T) \int^1_0 \biggl\{ Pq^N(\theta) - \bigg[ \frac{\partial C(q^N(\theta), \theta)}{\partial \theta} \frac{F^N(\theta)}{f^N(\theta)} + C(q^N(\theta), \theta)\bigg] \biggr\} dF^N(\theta)
	\end{align*}
	\begin{align}
		\text{s.t.} \quad 
		& \text{ for any } \theta, q^j(\theta) \text{ is non-increasing in } \theta, j \in \{N,T\}, \\
		& U^T(1) = \max_{\theta \in [0,1]} \int_{\theta}^{1} \bigg( \frac{\partial C(q^N(x), x)}{\partial x} - \frac{\partial C(q^T(x), x)}{\partial x} \bigg) dx.
	\end{align} 
		
	The principal's goal is to maximize the objective function by choosing quantities for each group, $q^T$ and $q^N$, and setting the utility for the most costly trait agent $U^T(1)$. In the objective, the term $\frac{\partial C(q^j(\theta), \theta)}{\partial \theta} \frac{F^j(\theta)}{f^j(\theta)} + C(q^j(\theta), \theta)$ is the \textit{virtual cost} of type $\theta$ in group $j$: it combines the direct production cost $C(q^j(\theta), \theta)$ with the informational rent markup $\frac{\partial C}{\partial \theta} \cdot \frac{F^j(\theta)}{f^j(\theta)}$, which the principal must concede to preserve the within-group incentive compatibility. Constraint (2) states that the utility for the most costly trait agent is determined by the maximum cumulative difference in the partial effect of type on cost between the two groups. To proceed, we use the Monotone Likelihood Ratio Property to order these two groups. We assume the following relation between $f^T$ and $f^N$:
	
	\begin{assu}[Monotone Likelihood Ratio Property (MLRP)]
		For all $\theta_1, \theta_2 \in [0,1]$ with $\theta_2 > \theta_1$, $$\frac{f^T(\theta_2)}{f^N(\theta_2)} \geq \frac{f^T(\theta_1)}{f^N(\theta_1)}.$$
	\end{assu}
	
	Equivalently, the likelihood ratio $f^T/f^N$ is non-decreasing in $\theta$. In our context, this assumption suggests that it is more likely to find an agent in the trait group if a higher cost $\theta$ is shown. This assumption is supported by the use of the Work Opportunity Tax Credit. Since the WOTC aims to encourage the principal to hire trait agents, this implies the trait agent may be discriminated against or have a disadvantage in the hiring process.
	
	Conversely, we consider the case where the likelihood ratio is ordered in the opposite direction, and we define it as the reverse MLRP:
	\begin{assu}[Reverse MLRP]
		For all $\theta_1, \theta_2 \in [0,1]$ with $\theta_2 > \theta_1$, $$\frac{f^T(\theta_2)}{f^N(\theta_2)} \leq \frac{f^T(\theta_1)}{f^N(\theta_1)}.$$
	\end{assu}
	
	Next, we define a commonly used assumption to impose regularity on the distributions. 
	
	\begin{assu}[Monotone Reverse Hazard Rate]
		Let the reverse hazard rate be $\frac{f(\theta)}{F(\theta)}$, where $F$ is a CDF and $f$ is the associated PDF. The monotone reverse hazard rate property holds if $\frac{f(\theta)}{F(\theta)}$ is non-increasing in $\theta \in [0,1]$.
	\end{assu}
	
	The reverse hazard rate is the instantaneous probability of an agent being drawn as type $\theta$, given that he is not drawn at any type $\theta' > \theta$, where $\theta'$ is a higher-cost type than $\theta$. We then require the reverse hazard rate to be monotone, which is not unusual in the standard mechanism design literature.
	
	To begin the analysis, we first consider the optimal contract for the most efficient agents in both groups, which are the agents with $\theta = 0$. Using the calculus of variations, we have the following proposition:
	\begin{prop}[No Distortion at the Top]\label{prop:startpoint}
		$\frac{\partial C(q^T(0),0)}{\partial q^T} = \frac{P}{1-\tau}$, and $\frac{\partial C(q^N(0),0)}{\partial q^N} = P$.
	\end{prop} 

	For the most efficient agent $\theta = 0$, there is no lower-cost type that could mimic him, so the principal faces no downward incentive pressure and can set quantities at the first-best level. For the non-trait group, the first-best condition is the standard equality of marginal cost and price. For the trait group, the tax credit $\tau$ effectively lowers the principal's net cost of production, raising the effective price to $\frac{P}{1-\tau} > P$; thus, the most efficient trait agent is induced to produce strictly more than the most efficient non-trait agent when $\tau > 0$. This gap widens as $\tau$ increases, reflecting the principal's incentive to extract greater output from the trait group when the government shares a larger part of the cost.
	
	Next, we consider a single-crossing condition on information rents of these two groups given a type $\theta$:
	\begin{assu}\label{assu:singlecrossinginfo}
		$\frac{F^T(\theta)}{f^T(\theta)} - \frac{F^N(\theta)}{f^N(\theta)}$ is non-decreasing in $\theta$.
	\end{assu}
	
	In words, the difference in the information rents required to incentivize the two groups is non-decreasing in $\theta$. 

	Next, we present our main theorem:
	\begin{thm}\label{thm:main}
		Let $0 \leq \tau < 1$. Assume the Monotone Reverse Hazard Rate holds for the relevant distributions or measures, and assume that $\frac{F^T(\theta)}{f^T(\theta)} - \frac{F^N(\theta)}{f^N(\theta)}$ is non-decreasing in $\theta$.
		\begin{enumerate}
			\item Under the MLRP, the optimal menu of contracts induces full separation up to a hiring cutoff $\tilde{\theta} \in (0,1]$. For $\theta \in [0, \tilde{\theta})$, the optimal quantities are the unique solution to 
			$$\Psi^T(q^T(\theta),\theta) = \frac{P}{1-\tau}, \qquad \Psi^N(q^N(\theta),\theta) = P,$$
			where $\Psi^j(q,\theta) := C_{q}(q,\theta) + \frac{F^j(\theta)}{f^j(\theta)}\,C_{q \theta}(q,\theta)$ is the virtual marginal cost for group $j$. Furthermore, $U^T(1) = 0$.

			\item Under the reverse MLRP, the optimal contract is characterized by a cutoff $\theta^* \in [0,1]$:
			\begin{itemize}
				\item \textit{Separating Region} ($0 \leq \theta < \theta^*$): quantities solve the virtual cost equations above. This region is non-empty when $\tau > 0$.
				\item \textit{Pooling Region} ($\theta^* \leq \theta \leq 1$): $q^T(\theta) = q^N(\theta) = q_p(\theta)$, the unique solution to
				$$P\hat{f}(\theta) = \widetilde{f}(\theta)\,C_{q}(q_p,\theta) + \widetilde{F}(\theta)\,C_{q \theta}(q_p,\theta),$$
				where $\hat{f}(\theta) := \pi_T f^T(\theta) + (1-\pi_T)f^N(\theta)$, $\widetilde{f}(\theta) := (1-\tau)\pi_T f^T(\theta) + (1-\pi_T)f^N(\theta)$, and $\widetilde{F}(\theta)$ is defined analogously.
			\end{itemize}
			When $\tau = 0$, $\theta^* = 0$ and the optimal contract is full pooling. When $\tau > 0$ and $\theta^* = 1$, the optimal contract induces full separation. Furthermore, $U^T(1) = 0$.
		\end{enumerate}
		The wage for each worker can be pinned down by the revenue equivalence result from Lemma~\ref{lemma:ICs}.
	\end{thm}
	
	We start the discussion of the first part of Theorem~\ref{thm:main}. Under MLRP, the likelihood ratio $f^T/f^N$ is non-decreasing, meaning higher-cost types are proportionally more concentrated in the trait group. A standard consequence is reverse hazard rate dominance: $\frac{F^T(\theta)}{f^T(\theta)} \leq \frac{F^N(\theta)}{f^N(\theta)}$ for all $\theta$; see \cite{krishna2009auction} for further discussion. This means the trait group requires a smaller information rent markup, so the virtual cost equation yields $q^T(\theta) > q^N(\theta)$ at every cost type. A trait agent who considers mimicking a non-trait agent would receive a lower quantity and thus a lower utility. Therefore, the across-group IC is slack and is automatically satisfied without imposing additional constraints. The principal can therefore optimize each group's contract independently, delivering full separation for all hired agents. The tax credit $\tau > 0$ reinforces this: it raises the effective price for the trait group from $P$ to $\frac{P}{1-\tau}$, further widening the output gap $q^T - q^N$ and strengthening trait agents' preference for their own contract.

	For the second part of Theorem~\ref{thm:main}, the likelihood ratio $f^T/f^N$ is non-increasing under reverse MLRP, so lower-cost types are relatively more concentrated in the trait group. The hazard rate ordering flips: $\frac{F^T(\theta)}{f^T(\theta)} \geq \frac{F^N(\theta)}{f^N(\theta)}$, meaning the trait group now faces a larger information rent markup. The virtual cost equations then yield $q^T(\theta) < q^N(\theta)$: trait agents receive more distorted contracts. A trait agent can gain by mimicking a non-trait agent and obtaining the higher-quantity, lower-distortion contract, violating the across-group IC under full separation. The single-crossing condition (Assumption~\ref{assu:singlecrossinginfo}) ensures this IC gap is monotone in $\theta$, so pooling is needed for a contiguous region of costly types. The optimal response is to pool all $\theta \geq \theta^*$ under the common contract $q_p(\theta)$, which equates $U^T = U^N$ in the pooled region and restores both the across-group IC and the within-group IC. The pooling quantity $q_p$ is determined by a blended virtual cost using the mixture distributions (or measures) $\widetilde{f}$ and $\widetilde{F}$, lying strictly between $q^T$ and $q^N$.\footnote{When $\tau > 0$, $\widetilde{F}$ is not a proper distribution as $\widetilde{F}(1) \neq 1 $. However, it is still well-defined as a measure because it is a combination of $F^T$ and $F^N$.} When $\tau > 0$, separation becomes profitable for efficient types $\theta < \theta^*$: the tax credit compensates the trait group enough that they are willing to accept a more distorted contract and reveal their identity.

	The key step in showing that pooling the groups is optimal under reverse MLRP is to verify that the principal's objective value is strictly decreasing in $U^T(1)$, so the optimum is $U^T(1) = 0$. The proof (see Appendix~\ref{app:pfs}) establishes that the marginal effect of $U^T(1)$ decomposes into three channels:
		\begin{enumerate}
			\item \textit{Transfer cost} ($-(1-\tau)\pi_T$): $U^T(1)$ is a pure transfer to the worst-cost trait agent, paid across the entire trait population and discounted by the tax credit.
			\item \textit{Production efficiency gain} ($\frac{(1-\tau)\hat{f}}{\widetilde{f}}|\Lambda|$): the benefit from tailoring quantities to each group's information rent over the pooling region, where $\Lambda(\theta) := \frac{\pi_T(1-\pi_T)(F^Nf^T - F^Tf^N)}{\hat{f}(\theta)}$ measures the disagreement between the two cost distributions.
			\item \textit{Tax credit loss} ($\Delta \leq 0$): when $\tau > 0$, separating the groups forfeits the ability to collect the tax credit differentially at the margin.
		\end{enumerate}
	Under reverse MLRP, channel (2) never compensates channel (1), and channel (3) only reinforces the case for pooling. The reason channel (2) never compensates channel (1) is a mismatch between who pays and who benefits. Raising $U^T(1)$ by a dollar raises every trait agent's utility by a dollar, because the envelope condition lifts the whole trait group's rent uniformly. The cost is therefore $(1-\tau)\pi_T$: a payment to the entire trait population, with no change in anyone's output. What the dollar buys is much narrower. The added slack in the disclosure constraint relaxes the pooling requirement, allowing the principal to tailor quantities to each group's information rent over a larger set of types. This tailoring is valuable only insofar as the two groups actually want different contracts, and $\Lambda(\theta)$ measures exactly this disagreement: it vanishes when the two cost distributions coincide, and it grows as they pull apart. Under reverse MLRP, $|\Lambda(\theta)|$ is non-decreasing in $\theta$ and is therefore largest at $\theta = 1$ (Lemma~\ref{lemma:Lambdanonpos}), which is exactly where the efficiency channel is evaluated. The bound $|\Lambda| < \pi_T$ therefore delivers a strong conclusion: even at the type where the disagreement between the groups is greatest, the tailoring gain is never worth the population-wide payment that purchases it. Since the bound depends only on the distributions, not on the cost function $C(q,\theta)$ or the tax credit $\tau$, the dominance of the transfer cost is a structural feature of the problem rather than an artifact of functional form.

    We summarize the proof of Theorem~\ref{thm:main} here. The first step is a change of variables that converts the problem into a standard one in the calculus of variations. The across-group incentive constraint becomes the pointwise state constraint, and we attach a multiplier function $\lambda(\theta) \geq 0$ with complementary slackness. Then, we can derive the Euler--Lagrange equations in which the multiplier's derivative equals, for each group, the derivative of that group's virtual-cost wedge. The structure of the optimum is related to the support of $\lambda$: wherever $\lambda(\theta) = 0$ the constraint is slack and each group's quantity solves its own virtual marginal cost equation, and wherever the constraint binds on an interval we can solve for a common quantity schedule whose first-order condition aggregates the two populations into the mixture measures $\tilde{f}$ and $\tilde{F}$. Separating and pooling regions thus emerge from complementary slackness, the junction $\theta^*$ is determined by continuity of the schedules together with the boundary conditions, and the distributional assumptions enter only at the final step, to sign the multiplier: MLRP guarantees $\lambda \equiv 0$ globally, while reverse MLRP guarantees the constraint binds on a contiguous upper region. Assumption~\ref{assu:singlecrossinginfo} enters at this same step, and its role is to shape the support of the multiplier: it requires $F^T(\theta)/f^T(\theta) - F^N(\theta)/f^N(\theta)$ to be non-decreasing, meaning the disclosure constraint passes from slack to binding at most once, so the binding set is a single interval and the optimal contract takes the clean cutoff form of Theorem~\ref{thm:main}.
    
    %Without Assumption~\ref{assu:singlecrossinginfo} the machinery above applies unchanged, but the support of $\lambda$ may be a union of intervals, and the optimal contract would alternate between separating and pooling bands rather than switching once at $\theta^*$.

	\section{Optimal Contract under the Quadratic Cost}\label{sec:quadratic}
	
	In this section, we restrict our attention to a quadratic cost function, $C(q, \theta) = \frac{1}{2}q^2 + \theta q$. We use this cost function to derive the explicit quantities and wages in the optimal contract. Notice that, under this quadratic cost function, the slope of the marginal cost is the same for all types $\theta$, while $\theta$ determines the intercept of the marginal cost. We separate the discussion of the MLRP case and the reverse MLRP case into two subsections.

	\subsection{Under MLRP}
	The first part of our Theorem~\ref{thm:main} can be rewritten as a corollary as follows:
	\begin{coro}
		Let $0 \leq \tau < 1$. Assume the Monotone Reverse Hazard Rate holds for both groups. Under MLRP, the optimal quantity and wage designed for each group are:
		\begin{align*}
			q^T(\theta) &= 
			\begin{cases}
				\frac{P}{1-\tau} - \theta - \frac{F^T(\theta)}{f^T(\theta)} & \text{if } \theta < \theta^{T*} \\
				0 & \text{if } \theta \geq \theta^{T*}
			\end{cases}, \\
			q^N(\theta) &= 
			\begin{cases}
				P - \theta - \frac{F^N(\theta)}{f^N(\theta)} & \text{if } \theta < \theta^{N*} \\
				0 & \text{if } \theta \geq \theta^{N*}
			\end{cases}, \\
			w^T \big( q^T(\theta),\theta \big) &= 
			\begin{cases}
				\int_{\theta}^{\theta^{T*}} q^T(x)dx + \frac{1}{2}(q^T(\theta))^2 + \theta q^T(\theta) & \text{if } \theta < \theta^{T*} \\
				0 & \text{if } \theta \geq \theta^{T*}
			\end{cases}, \\
			w^N \big( q^N(\theta),\theta \big) &= 
			\begin{cases}
				\int_{\theta}^{\theta^{N*}} q^N(x)dx + \frac{1}{2}(q^N(\theta))^2 + \theta q^N(\theta) & \text{if } \theta < \theta^{N*} \\
				0 & \text{if } \theta \geq \theta^{N*}
			\end{cases}.
		\end{align*}
		Furthermore, $\theta^{T*} = \frac{P}{1-\tau} - \frac{F^T (\theta^{T*})}{f^T (\theta^{T*})}$, $\theta^{N*} = P - \frac{F^N (\theta^{N*})}{f^N (\theta^{N*})}$, and $U^T(1) = 0$.
	\end{coro} 

	This result shows that the within-group separation comes from the Monotone Reverse Hazard Rate property, and the across-group separation comes from the Monotone Likelihood Ratio property; thus, the full separation is guaranteed. Following Appendix B from \cite{krishna2009auction}, we know the Monotone Likelihood Ratio Property implies the Reverse Hazard Rate Dominance; thus, we have a clear comparison of the information rents $\frac{F^T(\theta)}{f^T(\theta)}$ and $\frac{F^N(\theta)}{f^N(\theta)}$. As the contract for the non-trait group is more distorted than that for the trait group due to the MLRP assumption, trait agents receive more information rent by showing their evidence. Therefore, revealing the group identity is always weakly better than not revealing it for the trait agents. Furthermore, the most efficient trait agent produces more than the most efficient non-trait agent, and the cutoff type for the trait group depends on $\tau$. The full separation structure is preserved even as the tax credit increases. The tax credit works by encouraging all the trait agents to report their trait status. Therefore, when the tax credit increases, not only does the trait agent's utility increase, but the cutoff type $\theta^{T*}$ also increases, resulting in an increase in the hiring proportion of the trait group.

	We use a polynomial distribution example to demonstrate the effect of the tax credit:
	\begin{exmp}
		Let $F^N(\theta) = \theta^{0.3}$ and $F^T(\theta) = \theta^{0.7}$, and set $P = 0.75$ and $\pi_T = 0.3$. We will consider four levels of tax credit, $\tau \in \{0, 0.2, 0.4, 0.6\}$. Figure \ref{fig:optimal_quantities} illustrates the optimal contract design, where the full separation up to the cutoff type is represented by the blue line (non-trait agents) and the red line (trait agents). We see that when the tax credit increases, the production for the trait agents moves to the right, indicating a higher threshold.

		The two slopes of the optimal quantities are determined by the within-group IC, which means that if the slopes were flatter, the more efficient types would like to mimic the less efficient types within group. Since we are under MLRP, the blue line lies to the left of the red line, which shows that the trait agents receive higher utility than the non-trait agents with the same $\theta$.
		
		%Furthermore, we also compute the contract under the group-blind design, which is represented by the green line. We see that tax credit is not affecting the cutoff types, and the gap of the trait agent's cutoff and the group-blind cutoff is increasing drastically.\footnote{When tax credit is high, we see the group-blind contract design has area violating monotonicity, this is due to the violation of $\frac{\hat{f}(\theta)}{\tilde{f}(\theta)}$ being non-increasing in $\theta$. However, we can use the ironing procedure to make this area satisfies monotonicity; thus, the group-blind design will satisfy the monotonicity constraint.} 
		
		\begin{figure}[H]
			\centering
			\begin{subfigure}[b]{0.495\textwidth}
				\centering
				\includegraphics[width=\textwidth]{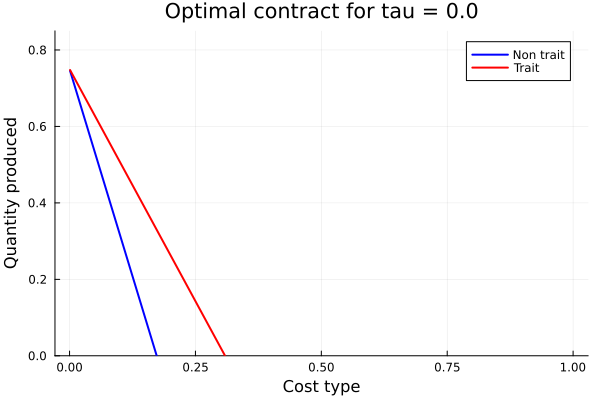}
			\end{subfigure}
			\hfill
			\begin{subfigure}[b]{0.495\textwidth}
				\centering
				\includegraphics[width=\textwidth]{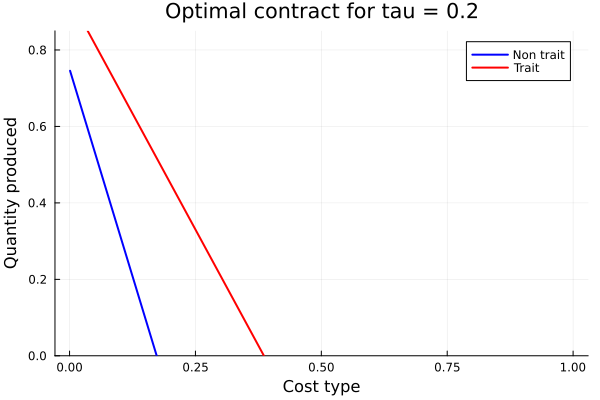}
			\end{subfigure}
			\vfill
			\begin{subfigure}[b]{0.495\textwidth}
				\centering
				\includegraphics[width=\textwidth]{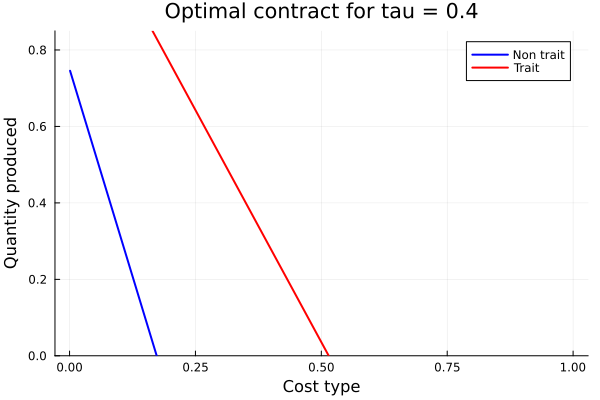}
			\end{subfigure}
			\hfill
			\begin{subfigure}[b]{0.495\textwidth}
				\centering
				\includegraphics[width=\textwidth]{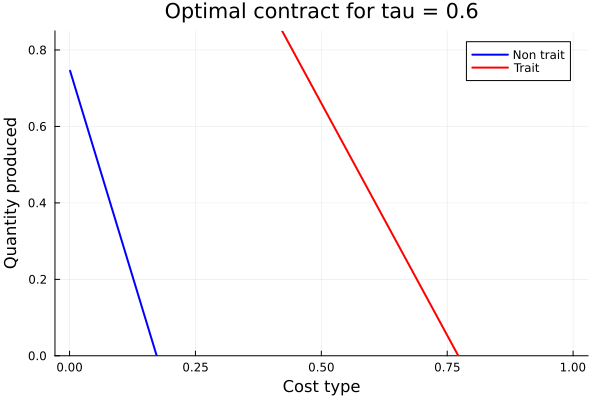}
			\end{subfigure}
			\caption{Optimal quantities under different tax credit levels under MLRP}
			\label{fig:optimal_quantities}
		\end{figure}
		
		Next, we consider the hiring probabilities under different levels of tax credit. Figure \ref{fig:hireprob} shows that when $\tau = 0$, the trait agents have a lower probability of being hired, even though their cutoff type is higher, as shown in Figure \ref{fig:optimal_quantities}. As $\tau$ increases, the hiring probability of trait agents equalizes with that of non-trait agents around $\tau = 0.4$.

		\begin{figure}[H]
			\centering
			\begin{subfigure}[b]{0.495\textwidth}
				\centering
				\includegraphics[width=\textwidth]{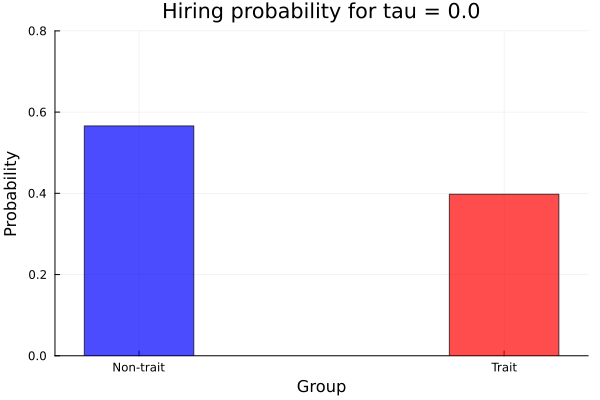}
			\end{subfigure}
			\hfill
			\begin{subfigure}[b]{0.495\textwidth}
				\centering
				\includegraphics[width=\textwidth]{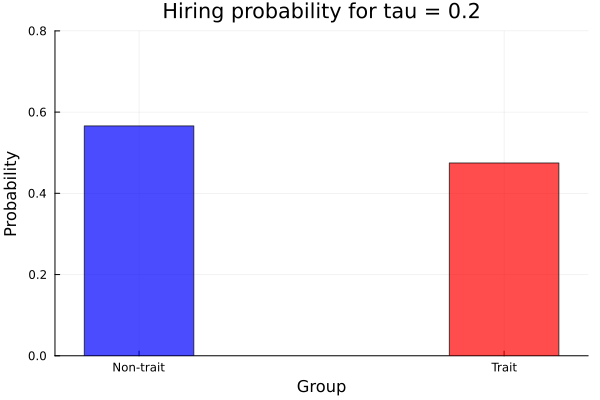}
			\end{subfigure}
			\vfill
			\begin{subfigure}[b]{0.495\textwidth}
				\centering
				\includegraphics[width=\textwidth]{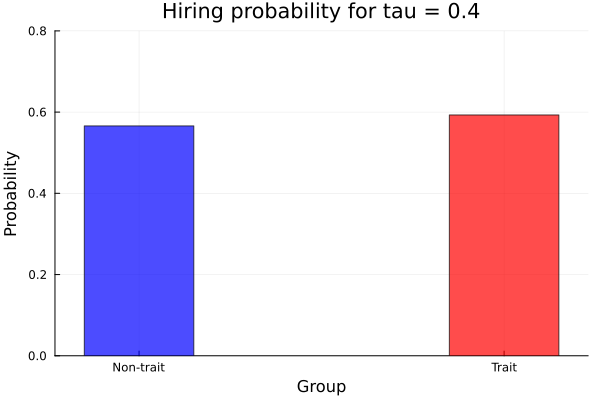}
			\end{subfigure}
			\hfill
			\begin{subfigure}[b]{0.495\textwidth}
				\centering
				\includegraphics[width=\textwidth]{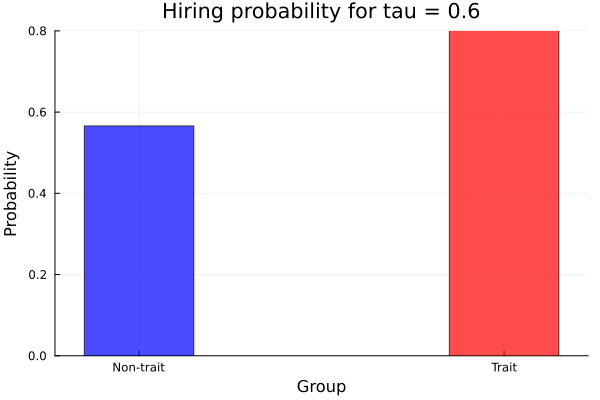}
			\end{subfigure}
			\caption{Hiring probabilities under different tax credit levels}
			\label{fig:hireprob}
		\end{figure}
	\end{exmp}

	\subsection{Under Reverse MLRP}
	The second part of our Theorem~\ref{thm:main} reduces to the following:
	\begin{coro}\label{coro:ReMLRP}
		 Let $0 \leq \tau < 1$. Assume the Monotone Reverse Hazard Rate Property, and that $\frac{F^T(\theta)}{f^T(\theta)} - \frac{F^N(\theta)}{f^N(\theta)}$ is non-decreasing in $\theta$.
		 
		 The optimal contract is characterized by a cutoff $\theta^* \in [0, 1]$ and the optimal production plan takes the following form:
		 \begin{itemize}
		 	\item[1.]  Separating Region ($0 \leq \theta < \theta^*$): The principal separates the groups to capitalize on the tax credit.
		 	$$q^T(\theta) = \frac{P}{1-\tau} - \theta - \frac{F^T(\theta)}{f^T(\theta)}, \quad q^N(\theta) = P - \theta - \frac{F^N(\theta)}{f^N(\theta)}.$$
		 	
		 	\item[2.] Pooling Region ($\theta^* \leq \theta \leq 1$): The principal pools the groups to minimize information rents.$$q^T(\theta) = q^N(\theta) = q_p(\theta) = P \frac{\hat{f}(\theta)}{\tilde{f}(\theta)} - \theta - \frac{\tilde{F}(\theta)}{\tilde{f}(\theta)},$$
		 \end{itemize}
		 
		 where $\hat{f} = \pi_T f^T + (1- \pi_T) f^N$ and $\widetilde{f} = (1-\tau) \pi_T f^T + (1- \pi_T) f^N$. $\hat{F}$ and $\widetilde{F}$ are defined similarly. 
		 
		 When $\tau = 0$, we have $\theta^* = 0$, and the optimal contract pools the groups over the entire domain.
	\end{coro}

	For Corollary~\ref{coro:ReMLRP}, when we consider $\tau = 0$, we are essentially comparing the distortion of production quantities and the transfer, $U^T(1)$. That is, we are studying the principal's objective by comparing the effect of distorting $q^N$ and $q^T$, and the effect of giving a pure transfer to the trait group. This corollary shows that the optimal design is to pool the two groups and offer them the pooling contract, and it implies that $U^T(1) = 0$ is the utility level designed for the most costly trait agent.
	
	When we set $U^T(1)$ high, the principal sets quantities closer to the optimal level of production as if the group identity were observed. Although the reverse MLRP assumption distorts the trait group's contract more (which gives an incentive for the trait agent to pretend to be a non-trait agent), in the end, the trait agent does not misreport because being in the trait group entitles the agent to receive this high transfer $U^T(1)$. When $U^T(1)$ starts decreasing, the principal treats the higher cost agents as if they were from the same group, and this pooling level of production distorts the quantity of the non-trait agents and redistributes the information rents toward the trait agents; thus, the level of information rent lies between $\frac{F^T(\theta)}{f^T(\theta)}$ and $\frac{F^N(\theta)}{f^N(\theta)}$. 

	When we consider $\tau > 0$, there must be a region where separating the trait agents from the non-trait agents is optimal for the principal, and the tax credit $\tau$ provides the incentive for separation. When $\tau$ increases, the region of separation increases until full separation is recovered. Even though reverse MLRP makes the trait group's production more distorted than the non-trait group's when the group is observed, the tax credit $\tau$ offsets the production distortion, thus reducing the effect of reverse MLRP.

	We use the same polynomial distribution to demonstrate the reverse MLRP case.
	\begin{exmp}
		Under the same polynomial distribution setting and under the reverse MLRP assumption, we first demonstrate how separating the groups leads to across-group IC violation when the tax credit $\tau$ is set to zero. The upper half of Figure~\ref{fig:Uprogession} shows four different contract designs, from the full separation, case (1), to pooling the groups, case (4). The lower half of Figure~\ref{fig:Uprogession} shows the types $\theta$ whose across-group IC is violated, and there is no violation in case (4).

		We present two types of semi-pooling, where case (2) pools the more efficient types and case (3) pools the more costly types. We see that with some level of pooling, the principal reduces the degree of across-group IC violation, and pooling the more costly types is better than pooling the more efficient types. 

		\begin{figure}[H]
			\centering
			\includegraphics[width=1.05\textwidth]{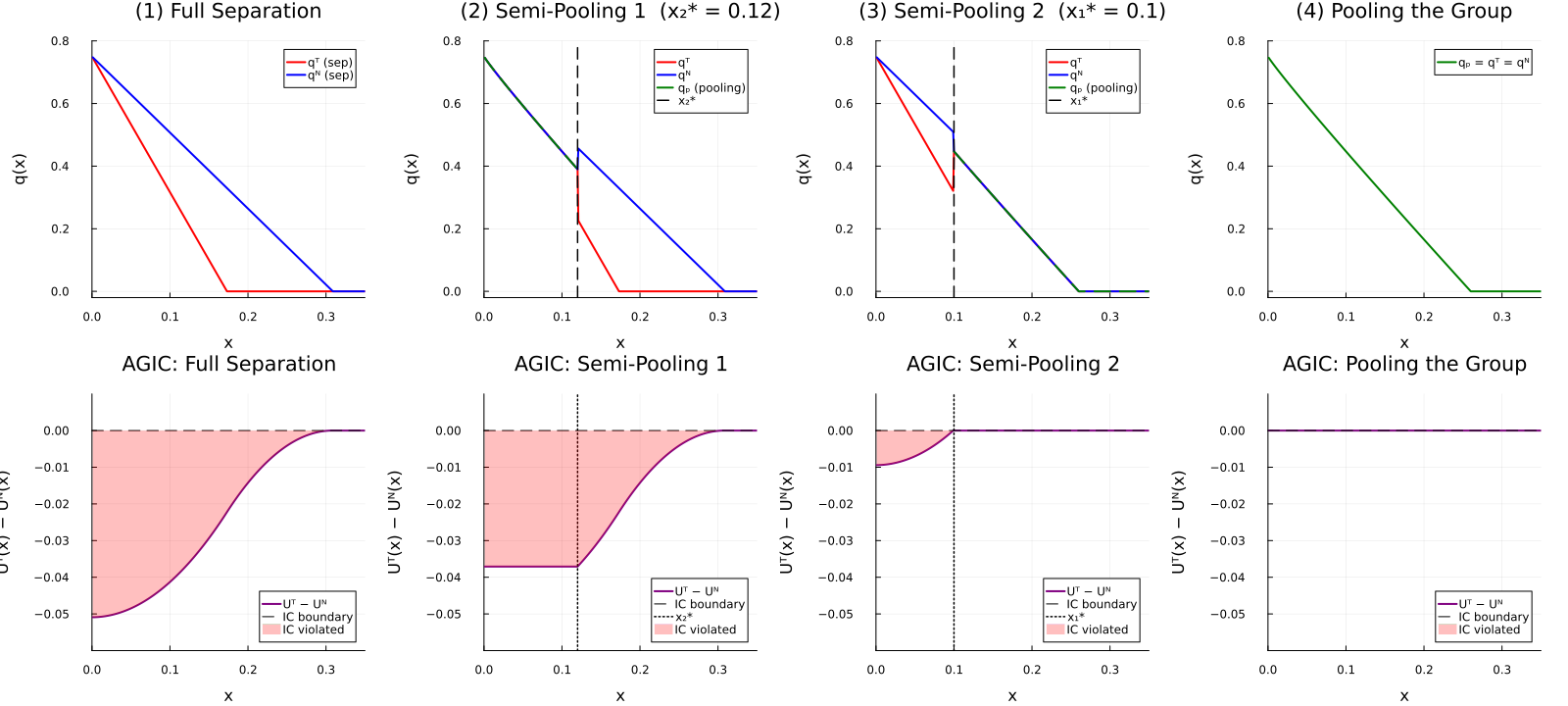}		
			\caption{IC comparison from full separation to pooling the groups.}
			\label{fig:Uprogession}	
		\end{figure}
		
		We next show the effect of tax credit $\tau$ on the optimal contract design. Figure~\ref{fig:quantities_ReMLRP} shows the effect of $\tau \in \{0, 0.2, 0.4, 0.6\}$. As we have discussed, pooling the groups (upper-left) is the optimal design when $\tau = 0$. As the tax credit $\tau$ increases, we have the more efficient trait agents being separated from the non-trait agents to receive the benefit of the tax credit. Because part of the information rents is now paid by the tax credit, the principal is willing to ask those efficient agents to produce at the optimal level as if the group were known.

		\begin{figure}[H]
			\centering
			\begin{subfigure}[b]{0.495\textwidth}
				\centering
				\includegraphics[width=\textwidth]{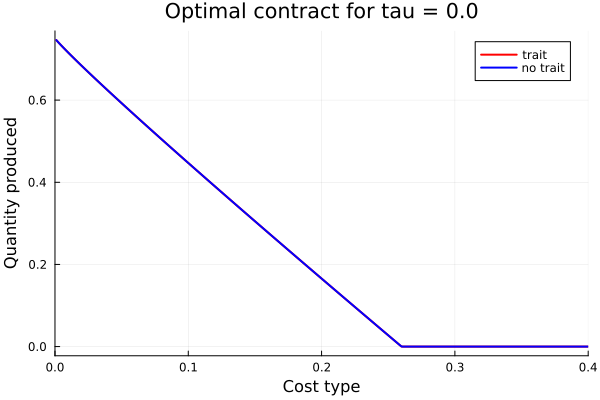}
			\end{subfigure}
			\hfill
			\begin{subfigure}[b]{0.495\textwidth}
				\centering
				\includegraphics[width=\textwidth]{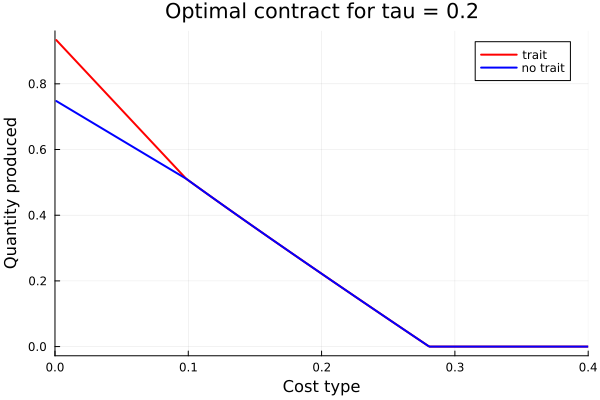}
			\end{subfigure}
			\vfill
			\begin{subfigure}[b]{0.495\textwidth}
				\centering
				\includegraphics[width=\textwidth]{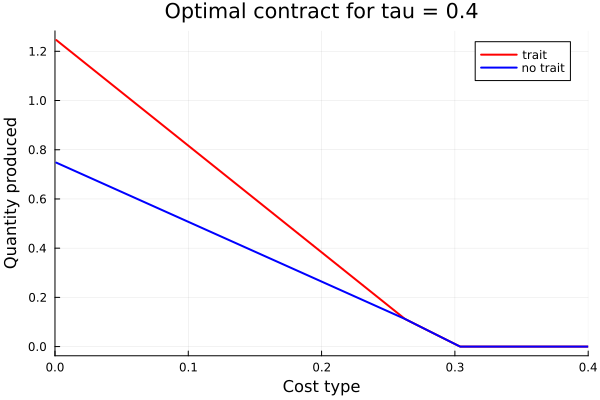}
			\end{subfigure}
			\hfill
			\begin{subfigure}[b]{0.495\textwidth}
				\centering
				\includegraphics[width=\textwidth]{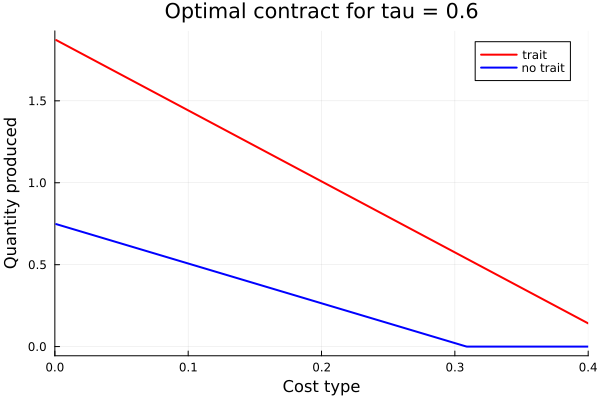}
			\end{subfigure}
			\caption{Optimal quantities under different tax credit levels under reverse MLRP}
			\label{fig:quantities_ReMLRP}
		\end{figure}
	\end{exmp} 

    \section{Policy Implication}\label{sec:policy}

    Theorem \ref{thm:main} gives a dichotomy. When the trait agents tend to have higher cost types (MLRP), the tax credit works through group-specific contracts: the principal hires the trait agents at cost types that she would reject among the non-trait agents, and she pays them more at any given cost type. This is affirmative action in the literal sense of differential treatment. When the trait agents tend to have lower cost types (reverse MLRP), the optimal response is the opposite. Over the pooling region, the principal offers a single group-blind menu, and group identity is ignored. Therefore, the same policy instrument endogenously produces either explicit group-targeting or de facto non-discrimination, and which regime arises depends on the stochastic ordering of the cost distributions, not on the policy itself.

    This dichotomy matters for how we interpret the empirical evidence on the WOTC. The evidence consistently shows that the subsidy produces little visible change in contracts.\footnote{The data coverage of these studies varies: the employer survey and agency data in \cite{gao2001wotc} cover California and Texas; \cite{jain2026limits} uses administrative data from Wisconsin, as does \cite{hamersma2008effects} cited at the end of this section; and the participation estimates in \cite{hamersma2003work} are based on national survey data.} Participating employers adjust their administrative practices to capture the credit, but fewer than ten percent of them report changing their hiring standards (\cite{gao2001wotc}). Certification rates among eligible new hires are low (\cite{hamersma2003work}). Moreover, \cite{jain2026limits} estimate precise null effects of subsidy eligibility on hiring, earnings, and retention, and they find that roughly ninety-six percent of subsidized hires are inframarginal. They attribute these null results to screening and organizational frictions. Our model supplies the complementary frictionless benchmark: under reverse MLRP, unchanged contracts are exactly what optimal screening looks like. Hence, the absence of a contractual response cannot by itself distinguish frictions from optimal group-blind screening. We emphasize that our purpose is not to defend the effectiveness of the program. Instead, we use the information structure of the WOTC, which this empirical literature documents in detail, as the leading case of a general contracting problem.

    Which ordering is the empirically relevant one? On one hand, MLRP is the premise of the program itself. The statute defines the target groups as individuals facing significant barriers to employment, and audit evidence confirms that a felony record or a long unemployment spell carries a real hiring penalty (\cite{pager2003mark}; \cite{agan2018ban}; \cite{kroft2013duration}). In the sense of \cite{milgrom1981good}, group membership is unfavorable news about the cost type. On the other hand, reverse MLRP is also plausible, because the distributions in our model are conditional on the pool that the principal actually faces. A firm screens applicants on observables before group status enters the contracting problem, and this prior screening can attenuate or even reverse the cost ordering in the population. If the trait applicants who survive the firm's screening are more favorably selected, relative to their group, than the non-trait applicants who survive, then the conditional distributions can satisfy reverse MLRP even when the population satisfies MLRP. This possibility is relevant for the WOTC because the firms that use it are disproportionately those that screen systematically: \cite{gao2001wotc} documents that three percent of participating firms account for roughly eighty-three percent of certifications. Therefore, the pools relevant to our model are exactly those in which selection is strongest. Furthermore, the target categories are not alike, so the ordering can run in opposite directions across them. For groups defined by disadvantage, such as ex-felons, membership is bad news about the cost type. For veterans, however, it is plausibly good news: military service brings training and a proven work history, which tends to lower the cost of employing a worker.
    
    Two features of our model that might appear to be conveniences are in fact faithful to the institutional design of the WOTC. First, the group-specific menus that arise under MLRP are lawful for the populations that the program targets. The designated categories, such as qualified veterans, ex-felons, and recipients of means-tested benefits, are not protected classes for the purpose of favorable treatment under federal employment law. This is unlike race or sex, which the WOTC deliberately does not target. When group-blindness is legally mandated, our model characterizes when the mandate has force. Under reverse MLRP, the mandate is slack on the pooling region because the unconstrained optimum already treats the two groups identically there. Under MLRP, the mandate binds wherever the separating contracts differ, and its shadow price is the difference between the value of the program in Theorem~\ref{thm:main} and the value of the optimal single-menu program.

    Second, the unidirectional disclosure constraint mirrors the certification process of the WOTC. The employer's claim rests on the information that the applicant supplies voluntarily during hiring. Most target statuses are stigmatized and unobservable without the worker's cooperation, and the audit evidence above shows that disclosure carries a real penalty, so concealment can be rational. At the same time, no worker can credibly claim a status he does not hold. This information structure has a substantive consequence for incidence. Because every trait agent retains the option to take the contract designed for the non-trait agents, the tax credit can never make him worse off through the contracting channel. The available evidence is consistent with this prediction: \cite{hamersma2008effects} finds that certified workers earn a short-run premium of roughly one-third of the credit's value, so the trait agents capture part of the subsidy as rent. This result stands in contrast to the literature following \cite{coate1993will}, in which affirmative action can harm its intended beneficiaries through equilibrium belief effects. In our model, the disclosure option is itself a form of protection, and it is exactly the feature that two-way verifiability would destroy.

	\section{Conclusion}\label{sec:conclusion}
	%%%%%%%%%%%%%%%%%%%%%%%%%%%%%%%%%%%%%%%%%%%%%%%%%%%%%%%%%%%%%%%%%%%%%%%%%%%%%%%%%%%%%%%
	%%%%%%%%%%%%%%%%%%%% Conclusion %%%%%%%%%%%%%%%%%%%%%%%%%%%%%%%%%%%%%%%%%%%%%%%%%%%%%%%%
	%%%%%%%%%%%%%%%%%%%%%%%%%%%%%%%%%%%%%%%%%%%%%%%%%%%%%%%%%%%%%%%%%%%%%%%%%%%%%%%%%%%%%%%
	In this paper, we study the optimal contract design when agents differ along two dimensions: a privately known cost type and membership in a protected group that qualifies the principal for a tax credit. Group membership is verifiable in only one direction: a trait agent can present evidence of his group identity or hide it, while a non-trait agent cannot claim it. Theorem~\ref{thm:main} shows that the optimal contracts depend on how the cost distributions of the two groups are ordered. Under MLRP, where the trait agents tend to draw higher cost types, the optimal contracts induce full separation among the hired agents, so each agent's contract choice reveals both his cost type and his group identity. In addition, the principal hires the trait agents up to a higher cost threshold than the non-trait agents, and the threshold rises with the tax credit. Under reverse MLRP, where the trait agents tend to draw lower cost types, the optimal contracts without the tax credit pool the groups and separate the agents by cost type alone. When the tax credit is introduced, the principal separates the groups only for the lower-cost agents, and increasing the tax credit enlarges the region of separation.

	These two regimes correspond to affirmative action and non-discrimination, and both arise endogenously from the principal's optimization rather than from regulation. As we discuss in Section~\ref{sec:policy}, this dichotomy matters for how the empirical evidence on the WOTC should be interpreted: under reverse MLRP, a tax credit that leaves the contracts unchanged is not evidence that the program fails, because pooling by cost type alone is exactly the optimal design over that region. On the methodological side, our analysis shows that a group dimension that affects payoffs only through the cost distribution, and can be misreported in only one direction, keeps the two-dimensional screening problem tractable under a general convex production cost. The calculus of variations then delivers a deterministic optimal contract without resorting to randomization.

	We close with two directions for future research. First, the structure we analyze, a one-way verifiable status coupled with a group-specific incentive, appears in settings well beyond the WOTC. Procurement set-asides, disability employment subsidies, and insurance markets in which clients can voluntarily disclose favorable test results all share this structure, and each application suggests a different ordering of the cost distributions. Second, a natural extension is to consider costly verification of evidence. Introducing costly verification would create an additional strategic layer in our model, where the trait agents must weigh the benefit of revealing their status against the cost of verification. This tradeoff could lead to partial pooling results. Such an extension would deepen our understanding of how diversity incentives operate in real-world settings where verification is rarely frictionless.

	%%%%%%%%%%%%%%%%%%%%%%%%%%%%%%%%%%%%%%%%%%%%%%%%%%%%%%%%%%%%%%%%%%%%%%%%%%%%%%%%%%%%%%
	%%%%%%% Reference %%%%%%%%%%%%%%%%%%%%%%%%%%%%%%%%%%%%%%%%%%%%%%%%%%%%%%%%%%%%%%%%%%
	%%%%%%%%%%%%%%%%%%%%%%%%%%%%%%%%%%%%%%%%%%%%%%%%%%%%%%%%%%%%%%%%%%%%%%%%%%%%%%%%%%%%%%
	\bibliography{ref_aa}
	\bibliographystyle{apalike}
	%\nocite{*} % To list all doc in the bibfile
	%%%%%%%%%%%%%%%%%%%%%%%%%%%%%%%%%%

	\newpage
	\begin{appendices}
		\section{Omitted Proofs}\label{app:pfs}
		\subsection{Proofs of Section \ref{sec:optimal}}

		\subsubsection*{Proof of Lemma \ref{lemma:withIC}}
		\begin{proof}
			($\Rightarrow$): Let $\theta_2 > \theta_1$. By within-group IC, we have
			\begin{align*}
				& w^j(\theta_1) - C \big(q^j(\theta_1), \theta_1 \big) \geq w^j(\theta_2) - C \big(q^j(\theta_2), \theta_1 \big) \\
				& w^j(\theta_2) - C \big(q^j(\theta_2), \theta_2 \big) \geq w^j(\theta_1) - C \big(q^j(\theta_1), \theta_2 \big) 
			\end{align*} 
			Rearranging terms, we get $C\big(q^j(\theta_2), \theta_2 \big) - C\big(q^j(\theta_2), \theta_1 \big) \leq C\big(q^j(\theta_1), \theta_2 \big) - C\big(q^j(\theta_1), \theta_1 \big)$. We then use $\frac{\partial^2 C}{\partial q \partial \theta } > 0$, which implies $C(q_2, \theta_2) - C(q_2, \theta_1) > C(q_1, \theta_2) - C(q_1, \theta_1) $ for any $q_2 > q_1$, to get $q^j(\theta_2) \leq q^j(\theta_1)$.
				
			Using the envelope condition, we have $\frac{d U^j(\cdot)}{d \theta} = - \frac{\partial C\big(q^j(\theta), \theta \big)}{\partial \theta}$. Using the Fundamental Theorem of Calculus and considering starting at type $\theta = 1$, we have 
			$$U^j(\theta) = U^j(1) + \int_{\theta}^{1} \frac{\partial C\big(q^j(x), x\big)}{\partial x} dx, $$
			where $U^j(\theta)$ is a simplification of $U(w^j(\theta), q^j(\theta); \theta)$.
				
			($\Leftarrow$): Suppose (1.) and (2.) hold. W.l.o.g., again let $\theta_2 > \theta_1$. We check if $U(w^j(\theta_2), q^j(\theta_2); \theta_2) \geq U(w^j(\theta_1), q^j(\theta_1); \theta_2)$:
			\begin{align*}
				& U^j(\theta_2) \geq U^j(\theta_1) - C\big(q^j(\theta_1), \theta_2 \big) + C\big(q^j(\theta_1), \theta_1 \big)\\
				&  \int_{\theta_2}^{1} \frac{\partial C\big(q^j(x), x\big)}{\partial x} dx \geq \int_{\theta_1}^{1} \frac{\partial C\big(q^j(x), x\big)}{\partial x} dx + C\big(q^j(\theta_1), \theta_1 \big) - C\big(q^j(\theta_1), \theta_2 \big) \\
				& -\int_{\theta_1}^{\theta_2} \frac{\partial C\big(q^j(x), x\big)}{\partial x} dx \geq - C\big(q^j(\theta_1), \theta_2 \big) + C\big(q^j(\theta_1), \theta_1 \big)\\
				& \int_{\theta_1}^{\theta_2} \frac{\partial C\big(q^j(x), x\big)}{\partial x} dx \leq C\big(q^j(\theta_1), \theta_2 \big) - C\big(q^j(\theta_1), \theta_1 \big) \\
				& \int_{\theta_1}^{\theta_2} \frac{\partial C\big(q^j(x), x\big)}{\partial x} dx \leq \int_{\theta_1}^{\theta_2} \frac{\partial C\big(q^j(\theta_1), x\big)}{\partial x} dx , \text{ by the Fundamental Theorem of Calculus}\\
				& \int_{\theta_1}^{\theta_2} \bigg[ \frac{\partial C\big(q^j(x), x\big)}{\partial x} - \frac{\partial C\big(q^j(\theta_1), x\big)}{\partial x} \bigg]dx \leq 0
			\end{align*}
			Using the single-crossing condition and monotonicity of $q^j$, we have $$\frac{\partial C\big(q^j(\theta), \theta\big)}{\partial \theta} - \frac{\partial C\big(q^j(\theta_1), \theta \big)}{\partial \theta} \leq 0,$$
			for all $\theta \in [\theta_1, \theta_2].$ Thus, within-group IC holds.
		\end{proof}

		\subsubsection*{Proof of Lemma \ref{lemma:acrossIC}}
		\begin{proof}
			By within-group IC constraints, we know that choosing contract $(w^j(\theta), q^j(\theta))$ gives an agent with cost $\theta$ utility $w^j(\theta) - C(q^j(\theta), \theta)$, where $j \in \{N,T\}$, and this is weakly better than choosing another contract $(w^j(\theta'), q^j(\theta'))$. This does not depend on whether the agent is a trait agent or a non-trait agent, as the utility function is the same.
		\end{proof}
		
		\subsubsection*{Proof of Lemma \ref{lemma:ICs}}
		\begin{proof}
			The proof that a set of contracts is within-group IC if and only if conditions (1.) and (2.) hold is in Lemma \ref{lemma:withIC}.
			
			($\Rightarrow$) Within-group IC and across-group IC together imply that 
			$$U\big( w^T(\theta), q^T(\theta); \theta \big) - U\big( w^N(\theta), q^N(\theta); \theta \big) \geq 0.$$
			By revenue equivalence,\footnote{Notice that $U\big( w^j(\theta), q^j(\theta); \theta \big) = w^j(\theta) - C(q^j(\theta), \theta).$} it becomes 
			$$U \big( w^T(\theta), q^T(\theta); \theta \big)|_{\theta = 1} + \int_{\theta}^{1} \frac{\partial C(q^T(x), x)}{\partial x} dx - U \big( w^N(\theta), q^N(\theta); \theta \big)|_{\theta = 1} - \int_{\theta}^{1} \frac{\partial C(q^N(x), x)}{\partial x} dx  \geq 0. $$
			
			($\Leftarrow$) Starting from (3.), we have $U\big( w^T(\theta), q^T(\theta); \theta \big) - U\big( w^N(\theta), q^N(\theta); \theta \big) \geq 0$, so the across-group IC holds under the same cost $\theta$. By Lemma \ref{lemma:acrossIC}, $U\big( w^N(\theta), q^N(\theta); \theta \big) \geq w^N(\theta') - C\big(q^N(\theta'), \theta\big)$ for any $\theta' \in [0,1]$. Therefore, across-group IC holds under any cost $\theta'$. 
		\end{proof}
		
		\subsubsection*{Proof of Lemma \ref{lemma:IR}}
		\begin{proof}
			$U^N(1) = 0$ follows from the standard argument, as the worst-cost non-trait agent does not have the incentive to mimic a non-trait agent with another cost, and he is not allowed to mimic any trait agent.
			
			As for the worst-cost trait agent, within-group IC and across-group IC together ensure that the utility for the worst-cost trait agent is nonnegative, which is $U^T(1) \geq U^N(1) = 0.$ Using Lemma~\ref{lemma:ICs} (3.), we have for any $\theta$, $\int_{\theta}^{1} \bigg( \frac{\partial C(q^N(x), x)}{\partial x} - \frac{\partial C(q^T(x), x)}{\partial x} \bigg) dx \leq U^T (1).$
			
			Let us define $I(\theta) = \int_{\theta}^{1} \bigg( \frac{\partial C(q^N(x), x)}{\partial x} - \frac{\partial C(q^T(x), x)}{\partial x} \bigg) dx$. We first show that $I(\theta)$ is continuous in $[0,1]$, and then the maximum, $\max_{\theta \in [0,1]} I(\theta)$, exists by the Weierstrass Extreme Value Theorem.

			To show $I(\theta)$ is continuous for $\theta \in [0,1]$, we first observe that $q^T(x)$ and $q^N(x)$ are both bounded on $[0,1]$ since they are non-increasing. Since $C(\cdot)$ is continuously differentiable and each $q^j$ is bounded, the partial derivative $\frac{\partial C(q^j(x), x)}{\partial x}$ is bounded on $[0,1]$ for $j \in \{N,T\}$. Thus, the integrand $g(x) := \frac{\partial C(q^N(x), x)}{\partial x} - \frac{\partial C(q^T(x), x)}{\partial x}$ is bounded: there exists a constant $M$ such that $|g(x)| \leq M$ for all $x \in [0,1]$. Next, consider a small change $\epsilon > 0$ from a $\theta$,
			$$|I(\theta+\epsilon) - I(\theta)| = \bigg|\int_{\theta+\epsilon}^{1} g(x)\, dx - \int_{\theta}^{1} g(x)\, dx\bigg| = \bigg|\int_{\theta}^{\theta+\epsilon} g(x)\, dx\bigg| \leq \int_{\theta}^{\theta+\epsilon} |g(x)|\, dx \leq M\epsilon.$$
			As $\epsilon \rightarrow 0$, we have $|I(\theta+\epsilon) - I(\theta)| \rightarrow 0$. The same argument applies for $\epsilon < 0$. Since $\theta$ is arbitrary, $I(\theta)$ is continuous on $[0,1]$. By the Weierstrass Extreme Value Theorem, $\max_{\theta \in [0,1]} I(\theta)$ exists.

			Thus far, we have shown Lemma~\ref{lemma:ICs} (3.) can be replaced by $U^T (1) \geq \max_{\theta \in [0,1]} I(\theta)$. Notice that the lower bound for $\max_{\theta \in [0,1]} I(\theta)$ is $0$ by choosing $\theta = 1$; therefore the IR constraint is satisfied. 
			
			Then, we show $U^T (1) = \max_{\theta \in [0,1]} I(\theta)$. We show by contradiction. Suppose $U^T (1) > \max_{\theta \in [0,1]} I(\theta)$. The principal can reduce $U^T (1)$ by a small amount, say $\eta$, such that $U^T (1) - \eta > \max_{\theta \in [0,1]} I(\theta)$. Let $\hat{U}^T (1)$ be denoted as $U^T (1) - \eta$, then $\hat{U}^T(\theta) = \hat{U}^T(1) + \int_{\theta}^{1} \frac{\partial C(q^T(x), x)}{\partial x}\, dx$ still satisfies the across-group IC, and we can reduce every trait agent's wage; therefore giving the principal a higher value.
		\end{proof}
		
		\subsubsection*{Proof of Proposition~\ref{prop:startpoint}}
		\begin{proof}
			We reformulate the problem using the calculus of variations. In this way, Proposition~\ref{prop:startpoint} becomes the boundary condition at $0$. We define $ U^T(1) = \underline{U}$, which is a constant, and we let $R^j(\theta) = \int_{\theta}^{1} \frac{\partial C(q^j(x), x)}{\partial x}dx$, where $j \in \{T, N\}$. $R^j(\theta)$ is the state variable of the problem. $r^j(\theta)$ represents the flow variable; we let $r^j(\theta) = \frac{\partial C(q^j(\theta), \theta)}{\partial \theta}$. Since we assume single-crossing ($\frac{\partial^2 C(q^j(\theta), \theta)}{\partial q^j \partial \theta} > 0$), $\frac{\partial C(\cdot)}{\partial \theta}$ is strictly increasing in $q^j$ and hence invertible. We define $\psi: \mathbb{R} \times [0,1] \to \mathbb{R}_+$ as
			$$\psi(r^j, \theta) := \left(\frac{\partial C(\cdot,\,\theta)}{\partial \theta}\right)^{-1}\!(r^j),$$
			i.e., $\psi(r^j, \theta)$ is the unique $q^j \geq 0$ satisfying $\frac{\partial C(q^j,\theta)}{\partial \theta} = r^j$.
			We then define $G: \mathbb{R} \times [0,1] \to \mathbb{R}$ by
			$$G(r^j, \theta) := C\!\left(\psi(r^j, \theta),\, \theta\right),$$
			which is the cost function re-expressed in terms of the flow variable $r$ and type $\theta$.

			Next, we describe our boundary conditions. For $\theta = 0$, which is the most efficient type in both groups, we have a constraint requiring that $\underline{U} + R^T(\theta = 0) \geq R^N(\theta = 0)$. For $\theta = 1$, we have fixed values for $R^T(\theta = 1)$ and $R^N(\theta = 1)$; furthermore, $R^T(\theta = 1) = R^N(\theta = 1) = 0$ by definition of $R^j(\theta)$. We can restate our constrained optimization problem as follows:			
				\begin{align*}
					\max_{R^T,R^N} \int^1_0 & \biggl\{\pi_T f^T(\theta)\bigg[ P \psi \big(r^T(\theta), \theta \big) - (1-\tau)\Big(\underline{U} + R^T(\theta) + G \big( r^T(\theta), \theta \big)\Big) \bigg]  + \\
					& (1-\pi_T) f^N(\theta) \bigg[ P \psi \big(r^N(\theta), \theta \big) - R^N(\theta) - G\big(r^N(\theta), \theta\big) \bigg] \biggr\}  d\theta
				\end{align*}
				\begin{align*}
					\text{s.t.} \quad 
					& \text{ } \psi(r^j(\theta), \theta) \text{ is non-increasing in } \theta, j \in \{N,T\}, \\
					& \text{ } \psi(r^j(\theta), \theta) \geq 0 \text{ for all } \theta, j \in \{N,T\}, \\
					& \underline{U} + R^T(\theta) \geq R^N(\theta), \text{ for all } \theta \in [0,1], \\
					& R^T(1) = R^N(1) = 0.
				\end{align*}
		
			We then apply a Lagrange multiplier $\lambda(\theta)$ to the constraint on state variables and consider the following modified problem:\footnote{We will again check the non-increasing constraint and the non-negativity constraint afterward.}
				
				\begin{align*}
					\mathcal{J} = \int^1_0 & \biggl\{\pi_T f^T(\theta)\bigg[ P \psi \big(r^T(\theta), \theta \big) - (1-\tau)\Big(\underline{U} + R^T(\theta) + G \big( r^T(\theta), \theta \big)\Big) \bigg]  + \\
					& (1-\pi_T) f^N(\theta) \bigg[ P \psi \big(r^N(\theta), \theta \big) - R^N(\theta) - G\big(r^N(\theta), \theta\big) \bigg] + \\
					& \lambda(\theta)[\underline{U} + R^T(\theta) - R^N(\theta)] \biggr\} d\theta
				\end{align*}
				
			We denote $\mathcal{L}$ as the inside part of the integral. We identify the boundary condition by considering the first variation of $\mathcal{J}$. Following Section 12 from \cite{kamien2012dynamic}, we denote ${R^j}'(\theta)  :=\frac{d R^j(\theta)}{d \theta} = - r^j(\theta) $. We rewrite $\mathcal{J}$ as
				\begin{align*}
					\mathcal{J} = \int^1_0 &\biggl\{ \pi_T f^T(\theta) \bigg[ P \psi \big(-{R^T}'(\theta), \theta \big) - (1-\tau)\Big(\underline{U} + R^T(\theta) + G \big(-{R^T}'(\theta), \theta \big)\Big) \bigg]  + \\
					& (1-\pi_T) f^N(\theta) \bigg[ P \psi \big(-{R^N}'(\theta), \theta \big) - R^N(\theta) - G \big(-{R^N}'(\theta), \theta \big) \bigg]  + \\
					& \lambda(\theta)[\underline{U} + R^T(\theta) - R^N(\theta)] \biggr\} d\theta
				\end{align*}

			Since we have an inequality constraint: $\underline{U} + R^T(\theta) - R^N(\theta) \geq 0$, for all $\theta \in [0,1]$, complementary slackness implies $\lambda(\theta)[\underline{U} + R^T(\theta) - R^N(\theta)] = 0$, thus we can drop it when considering the first variation of $\mathcal{J}$. Our integrand can be considered as a function: $\mathcal{L} \big(\theta, R^T(\theta), {R^T}'(\theta), R^N(\theta), {R^N}'(\theta) \big)$. Let $\mathcal{J}^*$ be the maximum objective value that can be achieved, and we consider the integrand being $\mathcal{L} \big(\theta, R^{T*}(\theta), {R^T}'^*(\theta), R^{N*}(\theta), {R^N}'^*(\theta) \big)$.
				\begin{align*}
					\mathcal{J} - \mathcal{J}^* = \int^1_0 \bigg[ \mathcal{L} \big(\theta, R^T(\theta), {R^T}'(\theta), R^N(\theta), {R^N}'(\theta) \big) - \mathcal{L} \big(\theta, R^{T*}(\theta), {R^T}'^*(\theta), R^{N*}(\theta), {R^N}'^*(\theta) \big) \bigg] d\theta.
				\end{align*}
			Since $\mathcal{J}^*$ is the maximum, $\mathcal{J} - \mathcal{J}^*$ has to be non-positive. We then use Taylor's Theorem around $\big(\theta, R^{T*}, {R^T}'^*, R^{N*}, {R^N}'^* \big)$:
				\begin{align*}
					&\mathcal{J} - \mathcal{J}^* = \\
					&\int^1_0 \bigg[ \big(R^T - R^{T*} \big) \frac{\partial \mathcal{L}^*}{\partial R^T} + \big({R^T}' - {R^T}'^* \big) \frac{\partial \mathcal{L}^*}{\partial {R^T}' } + \big(R^N - R^{N*} \big) \frac{\partial \mathcal{L}^*}{\partial R^N} + \big({R^N}' - {R^N}'^* \big) \frac{\partial \mathcal{L}^*}{\partial {R^N}' }\bigg] d\theta + h.o.t
				\end{align*}	

			where $h.o.t$ refers to higher-order terms. Let $R^j - R^{j*} := h^j$ and ${R^j}' - {R^j}'^* := {h^j}'$, where $j \in \{T,N\}$. Notice that $\int_{0}^{1} {h^j}' \frac{\partial \mathcal{L}^*}{\partial {R^j}' } d \theta = -h^j(0) \frac{\partial \mathcal{L}^*}{\partial {R^j}' }|_{\theta = 0} - \int_{0}^{1} h^j \frac{d}{d \theta} \frac{\partial \mathcal{L}^*}{\partial {R^j}' } d\theta$ by using integration by parts. Then:

			\begin{align*}
				\mathcal{J} - \mathcal{J}^* = & -h^T(0) \frac{\partial \mathcal{L}^*}{\partial {R^T}' }|_{\theta = 0} - h^N(0) \frac{\partial \mathcal{L}^*}{\partial {R^N}' }|_{\theta = 0} \\
				& + \int_{0}^{1} h^T \bigg(  \frac{\partial \mathcal{L}^*}{\partial R^T} - \frac{d}{d \theta} \frac{\partial \mathcal{L}^*}{\partial {R^T}' } \bigg) d\theta + \int_{0}^{1} h^N \bigg(  \frac{\partial \mathcal{L}^*}{\partial R^N} - \frac{d}{d \theta} \frac{\partial \mathcal{L}^*}{\partial {R^N}' } \bigg) d\theta + h.o.t
			\end{align*}	

			We redefine $\mathcal{J} - \mathcal{J}^*$ as $\delta \mathcal{J}$. Dropping the $*$ for convenience and $h.o.t$ as we are interested in the first-ordered necessary condition, we have our first variation being:
			
			\begin{align*}
				\delta \mathcal{J} = & -h^T(0) \frac{\partial \mathcal{L}}{\partial {R^T}' }|_{\theta = 0} - h^N(0) \frac{\partial \mathcal{L}}{\partial {R^N}' }|_{\theta = 0} \\
				& + \int_{0}^{1} h^T \bigg(  \frac{\partial \mathcal{L}}{\partial R^T} - \frac{d}{d \theta} \frac{\partial \mathcal{L}}{\partial {R^T}' } \bigg) d\theta + \int_{0}^{1} h^N \bigg(  \frac{\partial \mathcal{L}}{\partial R^N} - \frac{d}{d \theta} \frac{\partial \mathcal{L}}{\partial {R^N}' } \bigg) d\theta
			\end{align*}	

			The two integrals in the second line relate to the Euler--Lagrange equation and will be $0$ under $\delta \mathcal{J}$. Thus, we can simplify $\delta \mathcal{J}$ to:

			\begin{align*}
				\delta \mathcal{J} 
				& =  -h^T(0) \frac{\partial \mathcal{L}}{\partial {R^T}' }|_{\theta = 0} - h^N(0) \frac{\partial \mathcal{L}}{\partial {R^N}' }|_{\theta = 0} \leq 0\\
				& \Rightarrow -h^T(0) \biggl\{\pi_T f^T \bigg[ P \frac{\partial \psi(-{R^T}', \theta)}{\partial {R^T}'} - (1-\tau)\frac{\partial G(-{R^T}', \theta)}{\partial {R^T}'} \bigg]\biggr\}|_{\theta = 0} \\
				&\qquad - h^N(0) \biggl\{(1-\pi_T) f^N \bigg[ P \frac{\partial \psi(-{R^N}', \theta)}{\partial {R^N}'} - \frac{\partial G(-{R^N}', \theta)}{\partial {R^N}'}\bigg]\biggr\}|_{\theta = 0} \leq 0
			\end{align*}

			Since $\psi \big(-{R^j}', \theta\big)$ is evaluated at $r^j = -{R^j}'$, differentiating $\frac{\partial C(\cdot)}{\partial\theta} = -{R^j}'$ with respect to ${R^j}'$ and applying the inverse function theorem gives:
			$$\frac{\partial \psi(\cdot)}{\partial {R^j}'} = \frac{-1}{\frac{\partial^2 C(\cdot)}{\partial q^j \,\partial \theta}}, \qquad \frac{\partial G(\cdot)}{\partial {R^j}'} = \frac{\partial C(\cdot)}{\partial q^j} \cdot \frac{\partial \psi(\cdot)}{\partial {R^j}'} = \frac{- \frac{\partial C(\cdot)}{\partial q^j}}{\frac{\partial^2 C(\cdot)}{\partial q^j\,\partial \theta}},$$
			where all derivatives of $C$ are evaluated at $\big(\psi({R^j}', \theta),\, \theta\big)$. Substituting into $\delta \mathcal{J}$:
			\begin{align*}
				\delta \mathcal{J} = & -h^T(0) \biggl\{\pi_T f^T \frac{(1-\tau)\frac{\partial C}{\partial q^T} - P}{\frac{\partial^2 C}{\partial q^T\,\partial \theta}}\biggr\}\bigg|_{\theta = 0} - h^N(0) \biggl\{(1-\pi_T) f^N \frac{\frac{\partial C}{\partial q^N} - P}{\frac{\partial^2 C}{\partial q^N\,\partial \theta}}\biggr\}\bigg|_{\theta = 0} \leq 0.
			\end{align*}
			
			Next, we consider how the across-group IC affects $\delta \mathcal{J}$. Recall that across-group IC means $\underline{U} + R^T(\theta) \geq R^N(\theta)$ for all $\theta \in [0,1]$. At $\theta = 0$, this inequality also needs to hold. Therefore, at the optimal path either $\underline{U} + R^{T*}(0) > R^{N*}(0)$ or $\underline{U} + R^{T*}(0) = R^{N*}(0)$. We let the non-binding case be case 1, and the binding case be case 2.
			
			We start with case 1. Taking any $R^T(0)$ and $R^N(0)$, we know $\underline{U} + R^{T}(0) \geq R^{N}(0) \Rightarrow \underline{U} + h^T(0) + R^{T*}(0) \geq h^N(0) + R^{N*}(0)$, since $h^j(0) = R^j(0) - R^{j*}(0)$. Rearranging terms, we get $\underline{U} + R^{T*}(0) - R^{N*}(0) + h^T(0) - h^N(0) \geq 0$. As we are considering the non-binding case, the sum of the first three terms are strictly positive, and $h^T(0)$ and $h^N(0)$ can be chosen freely as they are small perturbations. For $\delta\mathcal{J} \leq 0$ to hold for all free choices of $h^T(0)$ and $h^N(0)$, each coefficient must equal zero independently; therefore, we have
			\begin{align*}
				\frac{\partial C(q^T(0), 0)}{\partial q^T} = \frac{P}{1-\tau} \qquad \text{ and } \qquad \frac{\partial C(q^N(0), 0)}{\partial q^N} = P.
			\end{align*}			
			% No distortion on the top condition.

			Next, we discuss case 2, where $\underline{U} + R^{T*}(0) = R^{N*}(0)$. Taking any $R^T(0)$ and $R^N(0)$, we again have $\underline{U} + R^{T*}(0) - R^{N*}(0) + h^T(0) - h^N(0) \geq 0 \Rightarrow h^T(0) \geq h^N(0)$, as the sum of the first three terms become zero in the binding case. This result implies that the perturbation for group $T$ has to be weakly larger than the perturbation for group $N$ for any perturbation applied to them. In other words, if we deviate from the optimal path, the utility gain (loss) for the most efficient trait agent needs to be no smaller (no larger) than the utility gain (loss) for the most efficient non-trait agent. 
			
			We then consider a function $g(\theta) = \underline{U} + R^T(\theta) - R^{N}(\theta)$. Notice that $g(\theta) \geq 0$ as this function captures the across-group IC at each $\theta \in [0,1]$. When the boundary is binding, $g(0) = 0$, which is a local minimum, and we consider the limit from the right: %\footnote{I need $q^T$ and $q^N$ being continuous at 0}
			\begin{align*}
				g'(0^+) & = \lim_{\theta \rightarrow 0^+} \frac{g(\theta) - g(0)}{\theta - 0} \geq 0 \\
				&\Rightarrow \lim_{\theta \rightarrow 0^+} \frac{\underline{U} + R^T(\theta) - R^N(\theta)}{\theta} \geq 0, \text{ as } g(0) = 0 \\
				& \text{ by L'Hopital's Rule, } \\
				&\Rightarrow \lim_{\theta \rightarrow 0^+} -r^T(\theta) + r^N(\theta) \geq 0 \\
				&\Rightarrow -r^T(0) + r^N(0) \geq 0
			\end{align*}
			Therefore, we have a condition requiring $r^N(0) \geq r^T(0)$. 
			
			To summarize, we have to satisfy two conditions in the binding case, which are: (1.) $h^T(0) \geq h^N(0)$ and (2.) $r^N(0) \geq r^T(0)$. We then prove by contradiction to show that the binding case cannot be optimal. $\delta \mathcal{J} \leq 0 \Rightarrow -h^T(0)M_T \leq h^N(0)M_N$, where $M_T := \pi_T f^T(0) \frac{(1-\tau)\frac{\partial C(q^T(0), 0)}{\partial q^T} - P}{\frac{\partial^2 C(q^T(0), 0)}{\partial q^T\,\partial \theta}}$ and $M_N := (1-\pi_T) f^N(0) \frac{\frac{\partial C(q^N(0), 0)}{\partial q^N} - P}{\frac{\partial^2 C(q^N(0),0)}{\partial q^N\,\partial \theta}}$.
			
			Since $h^T(0)$ and $h^N(0)$ can be any perturbations satisfying $h^T(0) \geq h^N(0)$, we consider two sets of $h^T(0)$ and $h^N(0)$ that give us the necessary conditions for both $M_T$ and $M_N$:
			\begin{itemize}
				\item[1.] $h^T(0) \geq 0$ and $h^N(0) = 0$. We have $-h^T(0)M_T \leq 0 \Rightarrow M_T \geq 0 \Rightarrow \frac{P}{1-\tau} \leq \frac{\partial C(q^T(0), 0)}{\partial q^T}$.
				\item[2.] $h^T(0) = 0$ and $h^N(0) \leq 0$. We have $0 \leq h^N(0)M_N \Rightarrow M_N \leq 0 \Rightarrow P \geq \frac{\partial C(q^N(0), 0)}{\partial q^N}$.
			\end{itemize}

			From 1. and 2., we have $\frac{\partial C(q^N(0), 0)}{\partial q^N} \leq P < \frac{P}{1- \tau} \leq \frac{\partial C(q^T(0), 0)}{\partial q^T}$. Since $\frac{\partial^2 C}{\partial q^2} > 0$, this implies $q^T(0) > q^N(0)$. By single-crossing ($\frac{\partial^2 C}{\partial q \partial \theta} > 0$), $\frac{\partial C}{\partial \theta}$ is strictly increasing in $q$, so $r^T(0) = \frac{\partial C(q^T(0), 0)}{\partial \theta} > \frac{\partial C(q^N(0), 0)}{\partial \theta} = r^N(0)$. This violates $r^N(0) \geq r^T(0)$ as long as $\tau > 0$. Therefore, the boundary condition for the optimal path has to be the non-binding case where $\underline{U} + R^{T*}(0) > R^{N*}(0)$.
			
			When $\tau = 0$, items 1. and 2. give $\frac{\partial C(q^T(0),0)}{\partial q^T} \geq P \geq \frac{\partial C(q^N(0),0)}{\partial q^N}$, so $q^T(0) \geq q^N(0)$. Combined with condition (2) that $r^N(0) \geq r^T(0)$, and single-crossing ($\frac{\partial^2 C}{\partial q \partial \theta} > 0$), we must have $q^T(0) = q^N(0)$, and hence $\frac{\partial C(q^T(0),0)}{\partial q^T} = \frac{\partial C(q^N(0),0)}{\partial q^N} = P$.	
		\end{proof}
			
		\subsubsection*{Proof of Theorem \ref{thm:main}}
		\begin{proof}
			We prove the theorem by constructing the Euler--Lagrange Equations to obtain the candidate solutions; we then compare them and identify the one that maximizes the principal's objective.
		 
			\medskip
			\noindent \textbf{Euler--Lagrange Equations and the Candidate Solutions}
			
			Consider the Euler--Lagrange equations:
			\begin{align*}
				& \frac{\partial \mathcal{L}}{\partial R^T} - \frac{d}{d \theta} \frac{\partial \mathcal{L}}{\partial {R^T}' } = 0, \\
				& \frac{\partial \mathcal{L}}{\partial R^N} - \frac{d}{d \theta} \frac{\partial \mathcal{L}}{\partial {R^N}'} = 0.
			\end{align*}
			
			To begin, we apply integration by parts to the $R^j(\theta) f^j(\theta)$ terms in the objective. Since $R^j(1) = 0$ and $F^j(0) = 0$, we have 
			
			$$\int_0^1 R^j(\theta) f^j(\theta)\, d\theta = \int_0^1 r^j(\theta) F^j(\theta)\, d\theta, \quad j \in \{T, N\}.$$

			After this substitution, $R^j$ no longer appears explicitly in the integrand, so $\frac{\partial \mathcal{L}}{\partial R^T} = \lambda(\theta)$ and $\frac{\partial \mathcal{L}}{\partial R^N} = -\lambda(\theta)$.
			
			We denote $\frac{\partial C(q^j,\theta)}{\partial q}$ as $C_{q}(q^j, \theta)$ and $\frac{\partial^2 C(q^j,\theta)}{\partial q \partial\theta}$ as $C_{q\theta}(q^j, \theta)$, for $j \in \{N,T\}$. 
			The Euler--Lagrange equations then give:
				\begin{align}
    			\lambda(\theta) 
   				 &= \frac{d}{d\theta}\left[\pi_T f^T(\theta)\frac{(1-\tau)C_{q}(q^T,\theta) - P}{C_{q \theta}(q^T,\theta)} + (1-\tau)\pi_T F^T(\theta)\right], \\
    			-\lambda(\theta) 
    			&= \frac{d}{d\theta}\left[(1-\pi_T)f^N(\theta)\frac{C_{q}(q^N,\theta) - P}{C_{q \theta}(q^N,\theta)} + (1-\pi_T)F^N(\theta)\right].
				\end{align}

			Adding up, the left-hand side vanishes. Applying the Fundamental Theorem of Calculus and evaluating at $\theta = 0$ using the boundary conditions $C_{q}(q^T(0),0) = \frac{P}{1-\tau}$ and $C_{q}(q^N(0),0) = P$, together with $F^j(0) = 0$, pins down the constant term $K = 0$. We thus have our first observation:

			\begin{fact}\label{fact:equl}
				\begin{align*}
    			\pi_T f^T(\theta)\frac{(1-\tau)C_{q}(q^T,\theta) - P}{C_{q \theta}(q^T,\theta)} + (1-\tau)\pi_T F^T(\theta) = -\left[(1-\pi_T)f^N(\theta) \frac{C_{q}(q^N,\theta) - P}{C_{q \theta}(q^N,\theta)} + (1-\pi_T)F^N(\theta)\right],
				\end{align*}
				for all $\theta \in [0,1]$.
			\end{fact}

			The following two cases consider the Karush--Kuhn--Tucker (KKT) conditions: (1.) separating region ($\lambda(\theta) = 0$) and (2.) pooling region ($\lambda(\theta) > 0$).
			
			\medskip
			\noindent \textbf{Case 1: Separating Region ($\lambda(\theta) = 0$)}
			
			In a separating region, complementary slackness tells us that $\underline{U} + R^T(\theta) - R^N(\theta) > 0$ for all $\theta$ in this region, and we can go back to the Euler--Lagrange equations and apply the Fundamental Theorem of Calculus to get:
			\begin{align*}
				& C_1 =  \pi_T f^T(\theta)\frac{(1-\tau)C_{q}(q^T,\theta) - P}{C_{q \theta}(q^T,\theta)} + (1-\tau)\pi_T F^T(\theta) \\
				& C_2 = (1-\pi_T)f^N(\theta)\frac{C_{q}(q^N,\theta) - P}{C_{q \theta}(q^N,\theta)} + (1-\pi_T)F^N(\theta)
			\end{align*}

			where $C_1$ and $C_2$ are two constants. From Fact~\ref{fact:equl}, we know $C_1 = -C_2$. Rearranging each equation, $q^T(\theta)$ and $q^N(\theta)$ are implicitly defined by:
			\begin{align}
				\frac{P}{1-\tau} &= C_{q}(q^T,\theta) - \frac{C_1 - (1-\tau)\pi_T F^T(\theta)}{(1-\tau)\pi_T f^T(\theta)} C_{q \theta}(q^T,\theta), \\
				P &= C_{q}(q^N,\theta) + \frac{C_1 + (1-\pi_T)F^N(\theta)}{(1-\pi_T)f^N(\theta)} C_{q \theta}(q^N,\theta).
			\end{align}

			The single-crossing condition $C_{q \theta} > 0$ and $C_{qq} > 0$ together guarantee we can pin down $q^T(\theta)$ and $q^N(\theta)$ for each $\theta$ uniquely, up to the constant $C_1$.

			\medskip
			\noindent\textbf{Case 2: Pooling Region ($\lambda(\theta) > 0$)}
			
			In a pooling region, complementary slackness tells us that $\underline{U} + R^T(\theta) - R^N(\theta) = 0$ for all $\theta$ in this region. Taking the derivative with respect to $\theta$ further implies $r^T(\theta) = r^N(\theta)$, which by the definition of $r^j$ means $\frac{\partial C(q^T,\theta)}{\partial\theta} = \frac{\partial C(q^N,\theta)}{\partial\theta}$. Since $\frac{\partial^2 C}{\partial q\,\partial\theta} > 0$, $\frac{\partial C(\cdot,\theta)}{\partial\theta}$ is strictly increasing in $q$, so this implies $q^T(\theta) = q^N(\theta) =: q_p(\theta)$. That is, for all cost types $\theta$ in this region, the principal asks for the same production regardless of group identity.
			
			Substituting $q^T(\theta) = q^N(\theta) = q_p(\theta)$ into Fact~\ref{fact:equl}, and using the identity $\hat{f}(\theta) - \tau\pi_T f^T(\theta) = \widetilde{f}(\theta)$, the pooling quantity $q_p(\theta)$ is implicitly characterized by:
			$$\widetilde{f}(\theta)\cdot C_{q}(q_p,\theta) + \widetilde{F}(\theta)\cdot C_{q \theta}(q_p,\theta) = P\hat{f}(\theta),$$
			
			where $\hat{f}(\theta) := \pi_T f^T(\theta) + (1-\pi_T)f^N(\theta)$, $\widetilde{f}(\theta) := \pi_T(1-\tau)f^T(\theta) + (1-\pi_T)f^N(\theta)$, and $\widetilde{F}(\theta) := \pi_T(1-\tau)F^T(\theta) + (1-\pi_T)F^N(\theta)$.\footnote{$\widetilde{F}(\theta)$ is a weighted combination of $F^T$ and $F^N$ but is not a proper CDF when $\tau > 0$, since $(1-\tau)\pi_T$ and $(1-\pi_T)$ do not sum to one.} The single-crossing condition $C_{q\theta} > 0$ and convexity $C_{qq} > 0$ together guarantee that this equation uniquely pins down $q_p(\theta)$ for each $\theta$. 
			
			Notice that when $\tau > 0$, the pooling region cannot be the whole $[0,1]$, because the boundary condition requires $\theta = 0$ to be in the separating region. Therefore, we impose distributional assumptions to characterize the exact pooling and separating regions.

			\medskip
			\noindent\textbf{Optimality of the Objective Function}

			We first consider the optimal design when group identity is observed by the principal. From Case 1, evaluating $C_1$ at $\theta = 0$ using the boundary condition $C_{q}(q^T(0),0) = \frac{P}{1-\tau}$ and $F^T(0) = 0$ gives $C_1 = 0$. The separating quantities $q^T(\theta)$ and $q^N(\theta)$ are then implicitly defined by:
			\begin{align}
				\Psi^T(q^T,\theta) &:= C_{q}(q^T,\theta) + \frac{F^T(\theta)}{f^T(\theta)} C_{q \theta}(q^T,\theta) = \frac{P}{1-\tau}, \label{opt:T} \\
				\Psi^N(q^N,\theta) &:= C_{q}(q^N,\theta) + \frac{F^N(\theta)}{f^N(\theta)} C_{q \theta}(q^N,\theta) = P. \label{opt:N}
			\end{align}
			where $\Psi^j$ captures the virtual cost for group $j$. Since $\frac{\partial \Psi^j}{\partial q^j} > 0$, each equation uniquely pins down $q^j(\theta)$. At the same production level $q$, the virtual cost difference is:
			$$\Psi^T(q,\theta) - \Psi^N(q,\theta) = \left(\frac{F^T(\theta)}{f^T(\theta)} - \frac{F^N(\theta)}{f^N(\theta)}\right) C_{q\theta}(q,\theta).$$

			\medskip
			\noindent\textit{(a) Under MLRP}

			$\frac{F^T(\theta)}{f^T(\theta)} \leq \frac{F^N(\theta)}{f^N(\theta)}$ for all $\theta$, so $\Psi^T(q,\theta) \leq \Psi^N(q,\theta)$ at the same $q$. Since $\frac{P}{1-\tau} > P$ and $\Psi^j$ is strictly increasing in $q^j$, satisfying both \eqref{opt:T} and \eqref{opt:N} simultaneously requires $q^T(\theta) > q^N(\theta)$ for all $\theta \in [0,1]$.\footnote{We cannot have $q^T(\theta) \leq q^N(\theta)$ for all $\theta$ since the boundary conditions give $\Psi^T(q^T(0),0) = \frac{P}{1-\tau} > P = \Psi^N(q^N(0),0)$, which combined with $\frac{\partial\Psi^j}{\partial q^j} > 0$ requires $q^T(0) > q^N(0)$.}

			We now verify the across-group IC constraint holds with $\underline{U} = 0$. Since $q^T(\theta) > q^N(\theta)$ for all $\theta$, single-crossing gives $r^T(\theta) = \frac{\partial C(q^T,\theta)}{\partial\theta} > \frac{\partial C(q^N,\theta)}{\partial\theta} = r^N(\theta)$. Integrating from $\theta$ to $1$ gives $R^T(\theta) > R^N(\theta)$ for all $\theta \in [0,1)$, so $U^T(\theta) \geq U^N(\theta)$ holds everywhere. Therefore, the full separating contract with $\underline{U} = 0$ is optimal under MLRP.

			\medskip
			\noindent\textit{(b) Under Reverse MLRP}

			 $\frac{F^T(\theta)}{f^T(\theta)} \geq \frac{F^N(\theta)}{f^N(\theta)}$ for all $\theta$, so $\Psi^T(q,\theta) \geq \Psi^N(q,\theta)$ at the same $q$. The separating solution from \eqref{opt:T} and \eqref{opt:N} then satisfies $q^T(\theta) \leq q^N(\theta)$ for all $\theta > 0$, with equality only at $\theta = 0$. Full separation requires $\underline{U} > 0$ to satisfy the across-group IC, while pooling corresponds to $\underline{U} = 0$. To determine which is optimal, we study how the principal's value $\mathcal{J}$ responds to changes in $\underline{U}$. Define the following functions along the optimal path:
			\begin{align}
				\Xi^T(\theta) &:= \pi_T f^T(\theta)\frac{(1-\tau)C_{q}(q^T,\theta) - P}{C_{q \theta}(q^T,\theta)} + (1-\tau)\pi_T F^T(\theta), \label{def:Xi_T} \\
				\Xi^N(\theta) &:= (1-\pi_T)f^N(\theta)\frac{C_{q}(q^N,\theta) - P}{C_{q \theta}(q^N,\theta)} + (1-\pi_T)F^N(\theta). \label{def:Xi_N}
			\end{align}

			The Euler--Lagrange equations from the first variation of $\mathcal{J}$ with respect to $R^T$ and $R^N$ take the compact form:
			\begin{align}
				\lambda(\theta) = \frac{d\Xi^T}{d\theta}, \qquad -\lambda(\theta) = \frac{d\Xi^N}{d\theta}. \label{eq:EL_Xi}
			\end{align}

			From Fact~\ref{fact:equl}, we know $\Xi^T(\theta)+ \Xi^N(\theta) = 0$, for all $\theta \in [0,1]$.

			Before proceeding, we state our claim.
			\begin{claim}\label{claim:pooling_optimal}
				$\mathcal{J}(0) \geq \mathcal{J}(\underline{U})$ for all $\underline{U} \geq 0$ and all $\tau \in [0,1)$.
			\end{claim}

			\begin{proof}
				We show $\mathcal{J}'(\underline{U}) < 0$ for all $\underline{U} \geq 0$. Since $\underline{U}$ enters $\mathcal{J}$ linearly, the envelope theorem gives:
				\begin{equation}
					\mathcal{J}'(\underline{U}) = -(1-\tau)\pi_T + \int_0^1 \lambda(\theta)\,d\theta. \label{eq:envelope}
				\end{equation}

				By the Fundamental Theorem of Calculus and \eqref{eq:EL_Xi}:
				$$\int_0^1 \lambda(\theta)\,d\theta = \Xi^T(1) - \Xi^T(0).$$
				The boundary condition (Proposition~\ref{prop:startpoint}) gives $C_q(q^T(0), 0) = P/(1-\tau)$, so $(1-\tau)C_q(q^T(0),0) - P = 0$. Since $F^T(0) = 0$:
				$$\Xi^T(0) = \pi_T f^T(0) \cdot \frac{0}{C_{q\theta}(q^T(0),0)} + 0 = 0.$$
				Therefore, $\mathcal{J}'(\underline{U}) = -(1-\tau)\pi_T + \Xi^T(1)$, and it remains to show $\Xi^T(1) < (1-\tau)\pi_T$.

				Since $\lambda \geq 0$, $\Xi^T$ is non-decreasing, and $\Xi^T(0) = 0$ gives $\Xi^T(1) \geq 0$. We bound $\Xi^T$ by considering the contract structure at $\theta = 1$.

				\textit{At any pooling point} $\theta$ where $q^T(\theta) = q^N(\theta) = q_p(\theta)$: define $\alpha(\theta) := C_q(q_p,\theta)/C_{q\theta}(q_p,\theta)$. Using $\Xi^T(\theta) + \Xi^N(\theta) = 0$ (Fact~\ref{fact:equl}), the condition
				$$\tilde{f}(\theta)\,\alpha(\theta) - \hat{f}(\theta)\frac{P}{C_{q\theta}} + \tilde{F}(\theta) = 0$$
				determines $\alpha$, where $\tilde{f} = (1-\tau)\pi_T f^T + (1-\pi_T)f^N$ and $\tilde{F} = (1-\tau)\pi_T F^T + (1-\pi_T)F^N$. Substituting back into $\Xi^T$ and decomposing (using $(1-\tau)\hat{f} - \tilde{f} = -\tau(1-\pi_T)f^N$ and $F^T\tilde{f} - f^T\tilde{F} = (1-\pi_T)(F^Tf^N - f^TF^N)$):
				\begin{equation}
					\Xi^T(\theta) = \frac{(1-\tau)\hat{f}(\theta)}{\tilde{f}(\theta)}\,|\Lambda(\theta)| + \Delta(\theta), \label{eq:Jprime_full}
				\end{equation}
				where $\Delta(\theta) := -\frac{\pi_T(1-\pi_T)\tau\, f^T f^N\, P}{\tilde{f}\,C_{q\theta}} \leq 0$ and we used $\Lambda(\theta) \leq 0$ under reverse MLRP (see Lemma~\ref{lemma:Lambdanonpos}). 
				
				If $\theta = 1$ lies in a pooling region, substituting \eqref{eq:Jprime_full} at $\theta = 1$ into $\mathcal{J}'(\underline{U}) = -(1-\tau)\pi_T + \Xi^T(1)$ gives the full three-channel decomposition:
				\begin{equation}
					\mathcal{J}'(\underline{U}) = \underbrace{-(1-\tau)\pi_T}_{\text{transfer cost}} + \underbrace{\frac{(1-\tau)\hat{f}(1)}{\tilde{f}(1)}\,|\Lambda(1)|}_{\text{efficiency gain}} + \underbrace{\Delta(1)}_{\text{tax credit loss}\;\leq\;0}. \label{eq:Jprime_decomp}
				\end{equation}
				The case in which $\theta = 1$ lies in a separating region is treated at the end of the proof.

				Since $\Delta \leq 0$, it suffices to show $\frac{(1-\tau)\hat{f}}{\tilde{f}}|\Lambda| < (1-\tau)\pi_T$. Expanding:
				$$\frac{(1-\tau)\hat{f}}{\tilde{f}}|\Lambda| = \frac{(1-\tau)\pi_T(1-\pi_T)\bigl(F^Tf^N - F^Nf^T\bigr)}{\tilde{f}}.$$
				We need $(1-\pi_T)(F^Tf^N - F^Nf^T) < \tilde{f}$, i.e.,
				$$(1-\pi_T)f^N\underbrace{(F^T - 1)}_{\leq\, 0} \;<\; (1-\tau)\pi_T f^T + (1-\pi_T)\underbrace{F^Nf^T}_{\geq\, 0}.$$
				The left side is non-positive while the right side is strictly positive. Therefore, $\Xi^T(\theta) < (1-\tau)\pi_T$ at any pooling point.

				If $\theta = 1$ belongs to a pooling region, we are done. If $\theta = 1$ belongs to a separating region, $\Xi^T$ is constant in that region, equal to the value at its left endpoint. If the separating region extends to $\theta = 0$, then $\Xi^T(1) = \Xi^T(0) = 0 < (1-\tau)\pi_T$. Otherwise, the left endpoint is a preceding pooling point $\theta_2$, and $\Xi^T(1) = \Xi^T(\theta_2) < (1-\tau)\pi_T$ by the bound above.

				In all cases, $\Xi^T(1) < (1-\tau)\pi_T$, so $\mathcal{J}'(\underline{U}) < 0$ for all $\underline{U} \geq 0$ and all $\tau \in [0,1)$. Since $\mathcal{J}$ is strictly decreasing in $\underline{U}$, the optimum is $\underline{U} = 0$.
			\end{proof}

			\medskip
			\noindent\textit{Contract structure at $\underline{U} = 0$.}

			Having established $\underline{U} = 0$, we now characterize the structure of the optimal contract.

			\textbf{Case $\tau = 0$.} Under reverse MLRP, the separating solution satisfies $q^T(\theta) \leq q^N(\theta)$ for all $\theta$, so $r^T(\theta) \leq r^N(\theta)$ by single-crossing. At $\underline{U} = 0$, the across-group IC requires $R^T(\theta) \geq R^N(\theta)$, which forces $R^T = R^N$ and hence $q^T = q^N$ everywhere: full pooling with $\theta^* = 0$.

			\textbf{Case $\tau > 0$.} The boundary condition (Proposition~\ref{prop:startpoint}) gives $C_q(q^T(0),0) = \frac{P}{1-\tau} > P = C_q(q^N(0),0)$, so $q^T(0) > q^N(0)$. Consider the separating quantities from \eqref{opt:T}--\eqref{opt:N} (with $C_1 = 0$). At any $\theta$ where $q^T(\theta) = q^N(\theta) = q$, the two virtual cost equations give:
			\begin{equation}
				\left(\frac{F^T(\theta)}{f^T(\theta)} - \frac{F^N(\theta)}{f^N(\theta)}\right) C_{q\theta}(q,\theta) = \frac{P\tau}{1-\tau}. \label{eq:crossing}
			\end{equation}
			At $\theta = 0$, the left side is zero (since $F^j(0) = 0$) while the right side is strictly positive, confirming $q^T(0) > q^N(0)$. Since $\frac{F^T(\theta)}{f^T(\theta)} - \frac{F^N(\theta)}{f^N(\theta)}$ is non-decreasing by assumption and $C_{q\theta} > 0$, the left side of \eqref{eq:crossing} is non-decreasing in $\theta$ (at a given $q$). Two cases arise:

			\begin{itemize}
				\item If the left side of \eqref{eq:crossing} remains below $\frac{P\tau}{1-\tau}$ for all $\theta \in [0,1]$ (i.e., the tax credit dominates), then $q^T(\theta) > q^N(\theta)$ throughout and the optimal contract is full separation ($\theta^* = 1$). The across-group IC is satisfied since $r^T > r^N$ everywhere implies $R^T > R^N$.

				\item Otherwise, there exists a crossing point $\theta^* \in (0,1)$ where $q^T(\theta^*) = q^N(\theta^*)$, determined by \eqref{eq:crossing}. For $\theta > \theta^*$, the separating solution would give $q^T(\theta) < q^N(\theta)$, making $r^T < r^N$ and violating the across-group IC near $\theta = 1$ (since $R^T(1) = R^N(1) = 0$ and $r^T < r^N$ implies $R^T(\theta) < R^N(\theta)$ for $\theta$ close to $1$).
			\end{itemize}

			In the second case, the principal pools on $[\theta^*, 1]$ with $q^T(\theta) = q^N(\theta) = q_p(\theta)$. On $[0, \theta^*)$, the separating quantities from \eqref{opt:T}--\eqref{opt:N} apply. The across-group IC at $\underline{U} = 0$ is verified: on $[\theta^*, 1]$, $r^T = r^N$, so $R^T(\theta) = R^N(\theta)$; on $[0, \theta^*)$, $q^T > q^N$ implies $r^T > r^N$, so $R^T(\theta) - R^N(\theta) = \int_\theta^{\theta^*}(r^T - r^N)\,dx > 0$. Therefore the across-group IC holds everywhere.
		\end{proof}

		\subsubsection*{Technical Lemmas}		
		\begin{lemma}\label{lemma:Lambdanonpos}
			Define $\Lambda(\theta) := \frac{\pi_T(1-\pi_T)\bigl(F^N(\theta)f^T(\theta) - F^T(\theta)f^N(\theta)\bigr)}{\hat{f}(\theta)}$, where $\hat{f}(\theta) = \pi_T f^T(\theta) + (1-\pi_T)f^N(\theta)$. Under reverse MLRP, $\Lambda(\theta)$ is non-increasing and $\Lambda(\theta) \leq 0$ for all $\theta \in [0,1]$.
		\end{lemma}
		\begin{proof}
			Under reverse MLRP, $F^N(\theta)f^T(\theta) - F^T(\theta)f^N(\theta) \leq 0$, so $\Lambda(\theta) \leq 0$. Consider taking the derivative w.r.t. $\theta$:
				\begin{align*}
					\frac{d\Lambda(\theta)}{d \theta} & = \pi_T(1-\pi_T) \frac{\hat{f} \big(F^N f^{T'} - F^T f^{N'}\big) - \big(F^N f^T - F^T f^N\big)\hat{f}'}{\hat{f}^2} \\
					&= \pi_T(1-\pi_T) \frac{ \big(f^{T'}f^N - f^{N'}f^T\big) \big(\pi_TF^T + (1- \pi_T)F^N\big)}{\hat{f}^2}
				\end{align*}
			Next, notice that we have $\frac{f^T(\theta)}{f^N(\theta)}$ being non-increasing under reverse MLRP, which implies $\frac{f^Nf^{T'} - f^T f^{N'}}{(f^N)^2} \leq 0$; that is, $ f^Nf^{T'} \leq f^T f^{N'}$. Using this observation, we get:
				\begin{align*}
					\mathrm{Sign} \bigg( \frac{d\Lambda(\theta)}{d \theta} \bigg) & = \mathrm{Sign} \big(f^{T'}f^N - f^{N'}f^T\big) \leq 0
			\end{align*}
				So $\Lambda(\theta)$ is non-increasing w.r.t. $\theta$ and $\Lambda(\theta)$ reaches the minimum when $\theta = 1$.
		\end{proof}

		\begin{lemma}\label{lemma:nootherpooling}
			Let $\tau = 0$. If the principal pools at cost types $ \theta \in [\theta_2, 1]$, that is $q^T(\theta) = q^N(\theta) = q_p(\theta)$, then she also pools at any cost type $\theta < \theta_2$.
		\end{lemma}
		\begin{proof}
			Suppose not. The principal separates at cost type $\theta \in (\theta_1, \theta_2)$. Consider the following two cases: (1.) $\theta_1 = 0$, and (2.) $\theta_1 > 0$. 
			
			To begin, notice that once we pool the production for any $\theta' \in [\theta_2, 1]$, we have $\underline{U} = 0$.\footnote{If there is a separating region from $\theta_3$ to $1$ and $\theta_3 > \theta_2$, then $\underline{U}$ will not be $0$, but the logic of this proof will still hold.}

			For Case 1, we have pooling for the higher cost types (close to 1) and separating for the lower cost types (close to 0). As the separation region hits the boundary, the boundary condition from Proposition~\ref{prop:startpoint} pins down $C_1 = 0$ in the Euler--Lagrange equations. The separating quantities $q^T(\theta)$ and $q^N(\theta)$ on $[0, \theta_2)$ are then implicitly defined by:

			$$\Psi^T(q^T,\theta) = P, \qquad \Psi^N(q^N,\theta) = P.$$

			Under reverse MLRP, $\frac{F^T(\theta)}{f^T(\theta)} \geq \frac{F^N(\theta)}{f^N(\theta)}$, so $\Psi^T(q,\theta) \geq \Psi^N(q,\theta)$ at the same $q$. Since $\Psi^j$ is strictly increasing in $q$, this implies $q^T(\theta) < q^N(\theta)$ for all $\theta \in (0, \theta_2)$. By single-crossing ($C_{q\theta} > 0$), $\frac{\partial C(q^T,\theta)}{\partial \theta} < \frac{\partial C(q^N,\theta)}{\partial \theta}$, i.e., $r^T(\theta) < r^N(\theta)$. Since $q^T = q^N$ on $[\theta_2, 1]$, we have $r^T = r^N$ there. Thus:
			$$U^T(\theta) - U^N(\theta) = \int_{\theta}^{1} \bigg(\frac{\partial C(q^T(x),x)}{\partial x} - \frac{\partial C(q^N(x),x)}{\partial x}\bigg) dx < 0$$
			for all $\theta \in [0, \theta_2)$, violating the across-group IC constraint.

			For Case 2, we have pooling for the higher cost types (close to 1), separating for the middle cost types until $\theta_1$, and pooling again for the lower cost types (close to 0). Since $\theta_1$ and $\theta_2$ belong to the pooling region, we have $\underline{U} + R^T(\theta_1) = R^N(\theta_1)$ and $\underline{U} + R^T(\theta_2) = R^N(\theta_2)$; in other words, the across-group IC binds at both endpoints. Since $\underline{U} = 0$, this implies $R^T(\theta_1) = R^N(\theta_1)$ and $R^T(\theta_2) = R^N(\theta_2)$. Subtracting gives $\int_{\theta_1}^{\theta_2} r^T(\theta)\, d\theta = \int_{\theta_1}^{\theta_2} r^N(\theta)\, d\theta$. As the region is separating from $\theta_1$ to $\theta_2$, we cannot have $r^T(\theta) = r^N(\theta)$ for all $\theta \in (\theta_1, \theta_2)$ (since separation requires $q^T \neq q^N$ and single-crossing gives $r^T \neq r^N$). For the integrals to be equal, $r^T(\theta)$ and $r^N(\theta)$ must intersect at some type in $(\theta_1, \theta_2)$, which by single-crossing ($C_{q\theta} > 0$) means $q^T$ and $q^N$ also intersect. But then neither $q^T(\theta) > q^N(\theta)$ nor $q^T(\theta) < q^N(\theta)$ holds throughout the separating region, contradicting the structure of the Euler--Lagrange solution which requires a consistent ordering within a single separating region.

			From these two cases, we conclude that more complicated partial pooling configurations are not candidates for the optimal contract.
		\end{proof}

%%%%%%%%%%%%%%%%%% Generalizing Cost Function %%%%%%%%%%%%%%%%%%%%%%%%%%%%%%%%%%%%%%%%			
	\section{Standard Method for Results under MLRP}\label{app:cost}
		\begin{thm}
			Assume both the Monotone Likelihood Ratio Property and the Monotone Reverse Hazard Rate property, and assume the cost function satisfies $\frac{\partial^3 C(q, \theta)}{\partial q^2 \partial \theta} \leq 0$ and $\frac{\partial^3 C(q, \theta)}{\partial q \partial \theta^2} \geq 0$.\footnote{Both conditions hold with equality for the quadratic cost function in Section~\ref{sec:quadratic}, where all third derivatives vanish.} The optimal contract design induces full separation up to a cutoff type $\theta^*$.
		\end{thm}
		\begin{proof}
			We proceed by solving the principal's problem without both the constraints, and checking if the solution satisfies the constraints. Ignoring both constraints, the principal optimally sets $U^T(1) = 0$.
			
			By taking the first order condition with respect to $q^T$ and $q^N$, we get:
			\begin{align*}
				& [q^T]: P = (1-\tau)\bigg[ \frac{\partial C(q^T(\theta), \theta)}{\partial q^T} + \frac{\partial^2 C(q^T(\theta), \theta)}{\partial q^T \partial \theta} \frac{F^T(\theta)}{f^T(\theta)} \bigg], \\
				& [q^N]: P =  \frac{\partial C(q^N(\theta), \theta)}{\partial q^N} + \frac{\partial^2 C(q^N(\theta), \theta)}{\partial q^N \partial \theta} \frac{F^N(\theta)}{f^N(\theta)}.
			\end{align*}
			The RHS is known as the virtual marginal cost for each group. The first term of the RHS is the marginal cost of producing $q^j$ given a cost type $\theta$, and the second term of the RHS is the information rent for cost type $\theta$, which consists of a cross derivative of the cost function and the inverse of the reverse hazard rate.
			
			We next show that given a production level $q$ and a cost type $\theta$, trait agent's virtual marginal cost is weakly lower than that of the non-trait agent. By subtracting the first order conditions, we get:
			$$-\tau \frac{\partial C(q(\theta), \theta)}{\partial q} + \frac{\partial^2 C(q(\theta), \theta)}{\partial q \partial \theta} \bigg( (1- \tau) \frac{F^T(\theta)}{f^T(\theta)} - \frac{F^N(\theta)}{f^N(\theta)}\bigg) \leq 0,$$
			since $\frac{\partial C(q(\theta), \theta)}{\partial q} > 0$, $\frac{\partial^2 C(q(\theta), \theta)}{\partial q \partial \theta} > 0$, and MLRP implies $\frac{F^N(\theta)}{f^N(\theta)} \geq \frac{F^T(\theta)}{f^T(\theta)}$. To equate two first order conditions, we must have $q^T(\theta) \geq q^N(\theta)$ for all $\theta$.
			
			Now we add both the conditions back and check if they are satisfied. Starting from the across-group IC condition, we have:
			\begin{align*}
				U^T(1) & = \max_{\theta \in [0,1]} \int^{1}_{\theta} \bigg( \frac{\partial C\big(q^N(x), x \big)}{\partial x} -  \frac{\partial C\big(q^T(x), x \big)}{\partial x} \bigg) dx \\
				& = \max_{\theta \in [0,1]} \int^{1}_{\theta} \bigg[ -\frac{P\tau}{1-\tau} + \bigg( \frac{F^T(x)}{f^T(x)} \frac{\partial^2 C(q^T(x), x)}{\partial q^T \partial x} - \frac{F^N(x)}{f^N(x)} \frac{\partial^2 C(q^N(x), x)}{\partial q^N \partial x}\bigg) \bigg] dx
			\end{align*} 
			
			$\frac{F^T(x)}{f^T(x)} \frac{\partial^2 C(q^T(x), x)}{\partial q^T \partial x} - \frac{F^N(x)}{f^N(x)} \frac{\partial^2 C(q^N(x), x)}{\partial q^N \partial x} \leq 0$ because of MLRP, $q^T(\theta) \geq q^N(\theta)$, and $\frac{\partial^3 C(q, \theta)}{\partial q^2 \partial \theta} \leq 0$. The integral is weakly negative for all $\theta \in [0,1]$; thus by choosing $\theta = 1$, we achieve the maximum value, which is $0$; thus, the across-group IC condition is satisfied.
			
			As for $q^j(\theta)$ being monotonic, we take the derivative with respect to $\theta$ for the two first order conditions. Considering the one for the trait agent, we have:
			\begin{align*}
				& [\theta]: -(1-\tau)\bigg[ \frac{\partial^2 C(q^T(\theta), \theta)}{\partial \theta \partial q^T} + \frac{F^T(\theta)}{f^T(\theta)} \frac{\partial^3 C(q^T(\theta), \theta)}{\partial q^T \partial \theta^2} + \frac{\partial^2 C(q^T(\theta), \theta)}{\partial q^T \partial \theta} \bigg( \frac{d}{d\theta}\frac{F^T(\theta)}{f^T(\theta)} \bigg) \bigg] \leq 0,
			\end{align*}
			as the first term inside the square bracket is positive by single-crossing, the second term is positive by the assumption $\frac{\partial^3 C(q, \theta)}{\partial q \partial \theta^2} \geq 0$, and the third term is positive by the monotone reverse hazard rate property together with single-crossing. Therefore, we have monotonically decreasing $q^T$. The argument is the same for $q^N$. As both production are monotonically decreasing in $\theta$, and $q^j(0) > 0$ for $j \in \{N,T\}$, if $q^j(\theta) < 0$ for some $\theta$, then by the Intermediate Value Theorem, there exists a type $\theta^*$ with $q^j(\theta^*) = 0$, which is the cutoff type. For all $\theta \geq \theta^*$, they receive a contract with zero wage and zero production; in other words, they are not hired.
		\end{proof}
	\end{appendices}

\end{document}